\documentclass[a4paper,12pt]{amsart}
\usepackage{amsmath,amssymb,amsthm}
\usepackage{amscd}
\usepackage{mathrsfs}
\usepackage[usenames,dvipsnames]{color}
\usepackage{graphicx}
\usepackage[all]{xy}
\usepackage{eucal}
\usepackage{extarrows} 
\usepackage{txfonts}      
\usepackage{tabularx}
\usepackage{tikz}
\usepackage{xparse}
\usepackage{colordvi}
\usepackage{multicol}
\usepackage{enumitem}
\usepackage[normalem]{ulem}
\usepackage{epsfig}
\usepackage{xspace}
\usepackage{stmaryrd}
\usepackage{tikz-cd}

\usepackage{ifpdf}
\ifpdf
 \usepackage[colorlinks,final,backref=page,hyperindex]{hyperref}
\else
 \usepackage[colorlinks,final,backref=page,hyperindex,hypertex]{hyperref}
\fi

\usepackage{graphicx}
\usepackage{epstopdf}
\usepackage{epsfig}
\usepackage{pdfsync}    

\usetikzlibrary{arrows}  
\usetikzlibrary{decorations.pathmorphing}  

\def\XXint#1#2#3{{\setbox0=\hbox{$#1{#2#3}{\int}$ }
\vcenter{\hbox{$#2#3$ }}\kern-.58\wd0}}

\def\XXsum#1#2#3{{\setbox0=\hbox{$#1{#2#3}{\sum}$ }
\vcenter{\hbox{$#2#3$ }}\kern-.51\wd0}}

\begin{document}

\newcommand\cutoffint{\mathop{-\hskip -4mm\int}\limits}
\newcommand\cutoffsum{\mathop{-\hskip -4mm\sum}\limits}
\newcommand\cutoffzeta{-\hskip -1.7mm\zeta}
\newcommand{\goth}[1]{\ensuremath{\mathfrak{#1}}}
\newcommand{\bbox}{\normalsize {}%
        \nolinebreak \hfill $\blacksquare$ \medbreak \par}
\newcommand{\simall}[2]{\underset{#1\rightarrow#2}{\sim}}

\newtheorem{theorem}{Theorem}[section]
\newtheorem{prop}[theorem]{Proposition}
\newtheorem{lemdefn}[theorem]{Lemma-Definition}
\newtheorem{propdefn}[theorem]{Proposition-Definition}
\newtheorem{lem}[theorem]{Lemma}
\newtheorem{lemma}[theorem]{Lemma}
\newtheorem{thm}[theorem]{Theorem}
\newtheorem{coro}[theorem]{Corollary}
\newtheorem{claim}[theorem]{Claim}
\newtheorem{outline}[theorem]{Outline}
\newtheorem{question}[theorem]{Problem}
\newtheorem{Goal}[theorem]{Goal}
\theoremstyle{definition}
\newtheorem{defn}[theorem]{Definition}
\newtheorem{rk}[theorem]{Remark}
\newtheorem{remark}[theorem]{Remark}
\newtheorem{ex}[theorem]{Example}
\newtheorem{coex}[theorem]{Counterexample}

\renewcommand{\theenumi}{{\it\roman{enumi}}}
\renewcommand{\theenumii}{\alpha{enumii}}

\newenvironment{thmenumerate}{\leavevmode\begin{enumerate}[leftmargin=1.5em]}{\end{enumerate}}

\setcounter{MaxMatrixCols}{20}

\newcommand{\nc}{\newcommand}

\nc{\delete}[1]{{}}


\nc{\mlabel}[1]{\label{#1}}  
\nc{\mcite}[1]{\cite{#1}}  
\nc{\mref}[1]{\ref{#1}}  
\nc{\meqref}[1]{\eqref{#1}}
\nc{\mbibitem}[1]{\bibitem{#1}} 

%
\delete{
\nc{\mcite}[1]{\cite{#1}{{\bf{{\ }(#1)}}}}  
\nc{\mlabel}[1]{\label{#1}  
	{\hfill \hspace{1cm}{\bf{{\ }\hfill(#1)}}}}
\nc{\mref}[1]{\ref{#1}{{\bf{{\ }(#1)}}}}  
\nc{\meqref}[1]{\eqref{#1}{{\bf{{\ }(#1)}}}}  
\nc{\mbibitem}[1]{\bibitem[\bf #1]{#1}} 
}


\nc{\mseva}{{\cale^Q_{\rm MS}}}

\nc{\wvec}[2]{{\scriptsize{\Big [ \!\!\begin{array}{c} #1 \\ #2 \end{array} \!\! \Big ]}}}

\nc{\lexle}{\le_{\mathrm{lex}}}

\nc{\mti}[1]{\widehat{#1}}

\nc{\fpower}{\calp_{\rm fin}}

\nc{\mzf}{multiple zeta fraction\xspace}
\nc{\mzfs}{multiple zeta fractions\xspace}

\nc{\pwvec}[2]{\Big(\wvec{#1}{#2}\Big)}

\nc{\fvec}[2]{\frakf\pwvec{#1}{#2}}
\nc{\name}[1]{{\bf #1}}

\nc{\scs}[1]{\scriptstyle{#1}}
\newfont{\scyr}{wncyr10 scaled 550}
\nc{\ssha}{\,\mbox{\bf \scyr X}\,}
\newfont{\bcyr}{wncyr10 scaled 1000}

\nc{\Sh}{\mathrm{Sh}}
\nc{\ug}{{U}}
\nc{\zsg}[1]{\widehat{#1}} \nc{\bs}{\bar{S}}
\nc{\spair}[2]{\Big[\begin{array}{c}\scs{#1} \\ \scs{#2} \end{array} \Big]}
\nc{\pfpair}[2]{\Big(\begin{array}{c}\scs{#1} \\ \scs{#2} \end{array} \Big)}

\nc{\supsp}{{\rm Supp}}
\nc{\depsp}{{\rm Dep}}
\nc{\pres}{\text{\rm p-res}}
\nc{\pord}{\text{\rm p-ord}}
\nc{\lexg}{>_{\rm lex}}
\nc{\perpq}{\perp^Q }
\nc{\qorth}{locality\xspace}			
\nc{\Qorth}{Locality\xspace}			
\nc{\msubalg}[1]{{\calm_{\Q+}}^{\perpq }\left\{#1\right\}}
\nc{\orth}{orthogonal\xspace}

\nc{\weak}{weak\xspace}
\nc{\wch}{{\rm wCh}}

\nc{\general}{ordered\xspace}
\nc{\gc}{L}

\nc{\thing}{Speer\xspace}
\nc{\tch}{{\rm Sp}}

\nc{\ssg}[1]{\overline{#1}}
\nc{\psg}[1]{\widetilde{#1}}
\nc{\shf}{{^{\ssha}}}
\nc{\lone}{_{\hskip -7.5pt 1}}
\nc{\qsshab}{{{\ssha\hspace{-2pt}_{\rho}}}\,}
\nc{\bsh}{{^{\qsshab}}}
\nc{\PF}{\mathbb{Q}\mathcal{F}}
\nc{\MZV}{\mrm{MZV}\xspace}
\nc{\MZVs}{\mrm{MZVs}\xspace}

\nc{\lexge}{\geq_{\mathrm{lex}}}
\nc{\lexgr}{>_{\mathrm{lex}}}
\nc{\sloc}{\top}

\nc{\Lyn}{\mathrm{Lyn}}

\nc{\ordfr}{\text{ordered fraction}\xspace}
\nc{\ordfrs}{\text{ordered fractions}\xspace}
\nc{\Ordercones}{Ordered cones\xspace}
\nc{\ordercones}{ordered cones\xspace}
\nc{\ordergerm}{ordered germ\xspace}
\nc{\ordergerms}{ordered germs\xspace}
\nc{\Ordergerms}{Ordered germs\xspace}

\nc{\mcone}[1]{C(#1)}
\nc{\conepair}[2]{C\spair{#1}{#2}}

\nc{\mpoly}[1]{{\calm_{\Q+}}^{\perpq }\left[#1\right]}
\nc{\Gal}{\mathrm{Gal}}
\nc{\mcalg}{\calg}
\nc{\calu}{\mathcal U}

\newcommand{\bottop}{\top\hspace{-0.8em}\bot}
\newcommand{\bbB}{\mathbb{B}}
\newcommand{\C}{\mathbb{C}}
\newcommand {\bbC}{\mathbb{C}}
\newcommand{\F}{\mathbb{F}}
\newcommand {\bbF}{{\mathbb{F}}}
\newcommand{\bbG}{\mathbb{G}}
\newcommand{\K}{\mathbb{K}}
\newcommand{\N}{\mathbb{N}}
\newcommand {\bbP}{{\mathbb{P}}}
\newcommand{\PP}{\mathbb{P}}
\newcommand{\Q}{\mathbb{Q}}
\newcommand{\R}{\mathbb{R}}
\newcommand {\bbR}{{\mathbb{R}}}
\newcommand{\T}{\mathbb{T}}
\newcommand{\bbU}{\mathbb{U}}
\newcommand{\W}{\mathbb{W}}
\newcommand {\bbW}{{\mathbb{W}}}
\newcommand{\Z}{\mathbb{Z}}
\newcommand{\fraka}{\mathfrak{a}}
\newcommand{\frakb}{\mathfrak{b}}
\newcommand {\frakc}{{\mathfrak {c}}}
\newcommand {\frakd}{{\mathfrak {d}}}
\newcommand {\frakf}{{\mathfrak {f}}}
\newcommand {\fraku}{{\mathfrak {u}}}
\newcommand {\fraks}{{\mathfrak {s}}}
\newcommand {\frakS}{{\mathfrak {S}}}
\newcommand {\cala}{{\mathcal {A}}}
\newcommand {\calb}{{\mathcal{B}}}
\newcommand {\calc}{{\mathcal {C}}}
\newcommand {\cald}{{\mathcal {D}}}
\newcommand {\cale}{{\mathcal {E}}}
\newcommand {\calf}{{\mathcal {F}}}
\newcommand {\calg}{{\mathcal {G}}}
\newcommand {\calh}{\mathcal{H}}
\newcommand {\cali}{\mathcal{I}}
\newcommand {\call}{{\mathcal {L}}}
\newcommand {\calm}{{\mathcal {M}}}
\newcommand {\calo}{{\mathcal {O}}}
\newcommand {\calp}{{\mathcal {P}}}
\newcommand {\calr}{{\mathcal {R}}}
\newcommand {\cals}{{\mathcal {S}}}
\newcommand {\calt}{{\mathcal {T}}}
\newcommand {\calv}{{\mathcal {V}}}
\newcommand {\calw}{{\mathcal {W}}}
\newcommand {\frakF}{{\mathfrak {F}}}

\nc{\vep}{\varepsilon}
\def \e {{\epsilon}}
\newcommand{\sy}[1]{{\color{purple}  #1}} 
\newcommand{\cy}[1]{{\color{cyan}  #1}}
\newcommand{\zb }[1]{{\color{blue}  #1}}
\newcommand{\li}[1]{{\color{red} #1}}
\newcommand{\lir}[1]{{\textcolor{red} {Li: #1}}}
\nc{\light}[1]{\textcolor{green} {(Out of context?) #1}}

\nc{\deff}{K}
\nc{\RR}{{\mathbb R}}
\nc{\ZZ}{{\mathbb Z}}
\nc{\QQ}{{\mathbb Q}}
\nc{\rd}{\mrm{rd}}
\nc{\monic}{{\mathfrak M}}
\nc{\coof}{L}           
\nc{\coef}{\bfk}        
\nc{\close}[1]{\overline{#1}}
\nc{\cl}{c}                 
\nc{\op}{o}                 
\nc{\io}{{io}}
\nc{\co}{\#}                 
\nc{\cc}{\mathcal{C}^c}      
\nc{\subd}{\mathrm{sub}}
\nc{\ccsp}{\calv^\cl}       
\nc{\ccrel}{\calw^\cl}      
\nc{\cccl}{\calh^\cl}       
\nc{\ocsp}{\calv^\op}       
\nc{\oc}{\mathcal{C}^o}      
\nc{\chenc}{\mathcal{CC}^o}  

\nc{\soc}{\mathcal{SC}^o}   
\nc{\doc}{\mathcal{DC}^o}    
\nc{\dsoc}{\mathcal{DSC}^o}  
\nc{\dosmc}{\mathcal{DTC}^o}    
\nc{\dfrsoc}{\mathcal{DFRSC}^o} 
\nc{\dfrcc}{\mathcal{DFRC}^c} 
\nc{\doch}{\mathcal{DCH}^o}    
\nc{\dcch}{\mathcal{DCH}^c}     
\nc{\dfrm}{\mathcal{DFRM}}   
\nc{\ocmzv}{\mathcal{CZV}^o} 
\nc{\oscmzv}{\mathcal{SZV}^o}  
\nc{\lzvset}{\mathcal{LZV}} 
\nc{\szvset}{\mathcal{SZV}} 
\nc{\spcc}{\mathcal{SpC}^c} 
\nc{\mcc}{\mathcal{MC}^c}    
\nc{\dcc}{\mathcal{DCC}^c}    
\nc{\cmzv}{\mathcal{DCZV}^c} 

\nc{\deco}[1]{\overline{#1}^o}  
\nc{\decc}[1]{\overline{#1}^c}  
\nc{\tn}{T}     
\nc{\dscc}{\mathcal{DSC}^c}  
\nc{\dsmc}{\mathcal{DMC}^c}  
\nc{\drc}{\mathcal{DRC}^c}  
\nc{\cdsmc}{\mathcal{CDMC}^c} 
\nc{\cs}{\mathcal{S}}   
\nc{\cm}{\mathcal{SM}}   
\nc{\ds}{\mathcal{DS}}  
\nc{\dm}{\mathcal{DM}}  

\nc{\smzv}{\mathcal{SZV}} 
\nc{\ocrel}{\calw^\op}      
\nc{\occl}{\calh^\op}       
\nc{\usp}{\{0\}}           
\nc{\cset}{\calc^\co}      
\nc{\csp}{\calv^\co}       
\nc{\crel}{\calw^\co}      
\nc{\ccl}{\calh^\co}       
\nc{\cone}[1]{\langle #1\rangle}
\nc{\ccone}[1]{\langle #1\rangle^\cl}
\nc{\ocone}[1]{\langle #1\rangle^\op}
\nc{\ocmzvset}{\mathcal{CZV}^o} 
\nc{\ccmzvset}{\mathcal{LZV}} 
\nc{\czvset}{\mathcal{CZV}}
\nc{\mzvset}{\mathcal{MZV}}
\nc{\dcset}{\calc^\co}      
\nc{\dcsp}{\calv^\co}       
\nc{\dcrel}{\calw^\co}      
\nc{\dccl}{\calh^\co}       
\nc{\dccset}{\calc^\cl}      
\nc{\dccsp}{\calv^\cl}       
\nc{\dccrel}{\calw^\cl}      
\nc{\dcccl}{\calh^\cl}       
\nc{\docrel}{\calw^\op}      
\nc{\doccl}{\calh^\op}       
\nc{\ccmzv}{\mathcal{CZV}^c}
\nc{\mzv}{\mathcal{MZV}}
\nc {\dtcp}{\mathcal {DTP}}

\nc{\mrm}[1]{{\rm #1}}
\nc{\Aut}{\mathrm{Aut}}
\nc{\depth}{{\mrm d}}
\nc{\id}{\mrm{id}}
\nc{\Id}{\mathrm{Id}}
\nc{\Irr}{\mathrm{Irr}}
\nc{\Span}{\mathrm{span}}
\nc{\mapped}{operated\xspace}
\nc{\Mapped}{Operated\xspace}
\newcommand{\redtext}[1]{{\textcolor{red}{#1}}}
\newcommand{\Hol}{\text{Hol}}
\newcommand{\Mer}{\text{Mer}}
\newcommand{\lin}{\text{lin}}
\nc{\ot}{\otimes}
\nc{\Hom}{\mathrm{Hom}}
\nc{\CS }{\mathcal{CS}}
\nc{\bfk}{{K}}
\nc{\lwords}{\calw}
\nc{\ltrees}{\calf}
\nc{\lpltrees}{\calp}
\nc{\Map}{\mathrm{Map}}
\nc{\rep}{\beta}
\nc{\free}[1]{\bar{#1}}
\nc{\OS}{\mathbf{OS}}
\nc{\OM}{\mathbf{OM}}
\nc{\OA}{\mathbf{OA}}
\nc{\based}{based\xspace}
\nc{\tforall}{\text{ for all }}
\nc{\hwp}{\widehat{P}^\calw}
\nc{\sha}{{\mbox{\cyr X}}}
\font\cyr=wncyr10 \font\cyrs=wncyr7
\nc{\Mor}{\mathrm{Mor}}
\def\lc{\lfloor}
\def\rc{\rfloor}
\nc{\oF}{{\overline{F}}}
\nc{\mge}{_{bu}\!\!\!\!{}}
\newcommand {\bfc}{{\bf {C}}}
\nc {\conefamilyc}{{\underline{{C}}}}
\newcommand{\loc}{locality\xspace}
\newcommand{\Loc}{Locality\xspace}

\nc{\supp}{\mathrm{Dep }}
\nc {\ordcone} {ordered cone\xspace}
\nc {\ordcones} {ordered cones\xspace}
\nc {\simplex}{simplex\xspace}
\nc{\QsubS}{\QQ\Pi^Q(\mti{\cals})}
\nc{\subS}{\Pi^Q(\mti{\cals})}
\nc{\wt}{weight\xspace}
\nc{\wts}{weights\xspace}

\title[Galois groups and
germs in several variables] {Locality Galois groups \\ of meromorphic germs in several variables
}

\author{Li Guo}
\address{Department of Mathematics and Computer Science,
	Rutgers University,
	Newark, NJ 07102, USA}
\email{liguo@rutgers.edu}

\author{Sylvie Paycha}
\address{Institute of Mathematics,
	University of Potsdam,
	D-14469 Potsdam, Germany}
\email{paycha@math.uni-potsdam.de}

\author{Bin Zhang}
\address{School of Mathematics, Sichuan University, Chengdu, 610064, China}
\email{zhangbin@scu.edu.cn}

\begin{abstract}  Meromorphic germs in  several variables with linear poles naturally arise in mathematics in various disguises. We investigate their rich structures under the prism of locality, including locality subalgebras, locality transformation groups and locality characters. The key technical tool is the dependence subspace for a meromorphic germ with which we define a locality orthogonal relation between two meromorphic germs. {  We describe the structure of  locality subalgebras  generated by   classes of meromorphic germs with certain types of poles}. We also define and determine {their  group  of locality transformations which fix the holomorphic germs and preserve  multivariable residues}, a group we call the locality Galois group.

We then specialise to two classes of meromorphic germs with prescribed types of nested poles, arising from multiple zeta functions in number theory and Feynman integrals in perturbative quantum field theory respectively. We show that they are locality polynomial subalgebras with locality polynomial bases given by the locality counterpart of Lyndon words.
This enables us to  explicitly describe their locality Galois group.
As an application, we {propose a mathematical interpretation} of Speer's analytic renormalisation for Feynman amplitudes.
We study a class of locality characters, called generalised evaluators after Speer. We show that the locality Galois group acts transitively on generalised evaluators by composition, thus providing a candidate for a renormalisation group in this multivariable approach.
\end{abstract}

\subjclass[2010]{
32A20, 
13B05, 
08A55,	
32A27, 
81T15, 
15A63, 
52C07, 
81T17, 
11M32, 
05A05   
}

\keywords{meromorphic germ, locality, Galois group, Lyndon word, renormalisation, evaluator, convex cone, renormalisation group}

\maketitle

\vspace{-1cm}
\tableofcontents

\vspace{-1cm}

\allowdisplaybreaks

\section{Introduction}

This work seeks to reveal the rich structure of meromorphic germs in several variables with linear poles, to describe subalgebras, to explore the structure of transformation group, to evaluate them at poles in a consistent manner, and to compare different evaluations. The locality framework developed by the authors {\cite{CGPZ1}}, appears to be well suited to achieve these goals. 

\subsection{From one to multivariable renormalisation}

Renormalisation is a procedure used to evaluate divergent expressions in various areas of  physics and mathematics, ranging from the classical instance of Feynman amplitudes in perturbative quantum field theory \mcite{t'H,t'HV} to multiple zeta functions at poles (see \mcite{GZ,MP} for example) and Todd functions for toric varieties  see, {e.g. \mcite{BV2,P}}.
In either case, a  preliminary step is a regularisation procedure, after which one can extract divergences and evaluate at the poles.
There is a great freedom in the choice of regularisation and the method of extracting divergences, specifically, a regularisation can involve one or multiple parameters. 

The algebraic structure underlying single parameter renormalisation
has attracted great interest  in mathematics since the groundbreaking work of Connes and Kreimer~\mcite{CK1,CK2} to tackle Feynman integrals.

In their approach, the combinatorics of the divergent expressions are organized into a connected Hopf algebra $\calh$, while the regularisations of the divergent expressions have their Laurent series expansions in the Rota-Baxter algebra $\calm(\C)=\C[z^{-1},z]]$,
characterised by its linear decomposition
\begin{equation} \mlabel{eq:splitMC} \calm(\C)= \calm_-(\C)\oplus \calm_+(\C),\end{equation}
into subalgebras $\calm_-(\C):=z^{-1}\C[z^{-1}], \calm_+(\C):= \C[[z]]$ and with the induced projection
\begin{equation} \mlabel{eq:proj1}
\pi_+:\calm(\C)\longrightarrow \calm_+(\C), \quad f(z)=\sum_{k=-K}^\infty a_k\, z^k\mapsto \sum_{k=0}^\infty a_k\, z^k.
\end{equation}
The regularisation map is enriched to an algebra homomorphism
\begin{equation} \mlabel{eq:phiMC} \phi:\calh \to \calm(\C),\end{equation}
which then factorises according to the algebraic Birkhoff factorisation, as the convolution product $\phi=\phi_-^{-1}\star \phi_+$ of a holomorphic part $\phi_+$ with values in the subalgebra $\calm_+(\C)$ and a polar part $\phi_-$ with values in $\calm_-(\C)$. The renormalised map $\phi_+$ which is then evaluated at the poles,
is built inductively with the recursion encoded in  the coproduct,  reflecting the celebrated  BPHZ procedure in perturbative quantum field theory \cite{BP,Hep,Zi}.

Extending the Connes-Kreimer approach to multiple parameter regularisations leads to an algebra homomorphism 
\[\phi:\calh\longrightarrow \calm(\C^\infty)\]
with values in $\calm(\C^\infty)$, the algebra of multivariable meromorphic germs at zero with linear poles for the Hopf algebra of convex polyhedral cones~\mcite{GPZ} and then later in a more general locality framework ~\mcite{CGPZ1,CGPZ2,CGPZ3}. This approach uses in an essential way a locality version of the  Rota-Baxter algebraic structure on $\calm(\C^\infty)$. 

\subsection{Locality for multivariable meromorphic germs} 
In physics, the principle of locality is a key feature of field theory which states that an object is influenced directly only by its immediate surroundings. We interpret locality more generally as certain binary relations ~{\mcite{CGPZ1}},  {enhance algebras to locality algebras  and call locality morphisms, the  morphisms  that preserve such locality relations. Its relation to causality in quantum field theory was discussed in \mcite{Re}.} 
In this paper the locality relation {on meromorphic germs} is  induced by an inner product $Q$ on the underlying vector space in the following way. A pair of   meromorphic germs lies in the graph of the locality relation if the linear spaces spanned by the  sets of variables they respectively depend on, called their {\bf  dependence subspaces}, are mutually orthogonal.This locality provides a natural splitting
$$\calm(\C^\infty)=\calm_+(\C^\infty)\oplus \calm_-^Q(\C^\infty)$$ into the subspace $\calm_+(\C^\infty)$ of holomorphic germs and a space $\calm_-^Q(\C^\infty)$ of what we call ``polar germs".
{While  $\calm_-^Q(\C^\infty)$ is  not a subalgebra of $\calm^Q(\C^\infty)$, it bares the remarkable property of being a {\bf locality ideal}.}

 {In applications ~\mcite{CGPZ2,CGPZ3}, we equip the Hopf algebra with a locality structure, turning it into a locality Hopf algebra, and $\phi$ becomes a locality morphism of algebras.} 
 {Thanks} to the fact that $\calm^Q(\C^\infty)$ is a locality ideal, not only the algebraic Birkhoff factorisation of $\phi$ can be recovered in the locality setting, {moreso,} the renormalisation procedure $\phi_+$ simplifies to the post composition $\pi_+^Q\circ \phi$ with the projection
$$ \pi_+^Q: \calm(\C^\infty) \to \calm_+(\C^\infty)$$ 
on the holomorphic part. {Its composition  ${\rm ev}_0\circ \pi_+^Q$ with the evaluation at zero ${\rm ev}_0:\calm_+(\C^\infty)\to \C$ on the resulting holomorphic germs, can be viewed as  a {\bf multivariable  minimal substraction scheme}.}

A prototype of this approach {was proposed in} the pioneering work of Speer~\mcite{Sp2,Sp4} on analytic renormalisation in quantum field theory. {This multivariable renormalisation method was later implemented}  for the baby model of Riemann integrals indexed by rooted trees in~\mcite{CGPZ3} {and further discussed in~\mcite{DZ}}. 

Thus while passing from the classical approach of one parameter renormalisation to the locality approach of multiple parameter renormalisation, the focus of analysis is shifted from the source Hopf algebra $\calh$ of $\phi$ to the target algebra $\phi(\calh)$ in $\calm(\C^\infty)$ discussed in this paper.

\subsection{Locality Galois groups and locality Lyndon words}
In practice, we consider locality subalgebras $\phi(\calh)\subset \cala\subset \calm(\C^\infty)$. { A linear} transformation $T$ on $\cala$ 
 induces   another  morphism
\[T\circ \phi: \calh\longrightarrow  \cala\subset\calm(\C^\infty),\]
from which we can again build a map
\[\pi_+^Q\circ T\circ \phi: \calh\longrightarrow  \calm_+(\C^\infty).\]
 
Locality isomorphisms of locality subalgebras $\cala\subset \calm(\C^\infty)$, which restrict to the identity map on $\calm_+(\C^\infty)$ form a group ${\rm Gal}^Q(\cala/\calm_+)$ we call {\bf locality Galois groups}. { It plays the role of a renormalisation group {\cite{CK2,CM}} in relating different renormalisations.}

Evaluating at zero by the map ${\rm ev}_0$, gives rise to a locality character \[{\mathcal E}_T:={\rm ev}_0\circ\pi_+^Q\circ T\circ \phi: \calh\longrightarrow \C,\]
which after Speer, we call  an evaluator, depending on the choice of $T$.

To describe the action of the locality Galois group, we give a careful study of the structure of locality subalgebras of $\calm(\C^\infty)$, defined by prescribed types of linear poles. We focus on Chen type poles which typically arise from multiple zeta functions and on the more general class of \thing fractions which arise in Feynman integrals.
We show that they both form {\bf \qorth polynomial algebras}, with a locality polynomial basis given a locality version of Lyndon words, we call {\bf locality Lyndon words}. For this purpose, we enhance to the locality setup, the realisation of shuffle product algebras as polynomial algebras on Lyndon words in the classical work of Chen-Fox-Lyndon and Radford~\cite{CFL,Ra}. 
As an application, transitivity is established for the action of the locality Galois groups on the generalised evaluators on these two classes of locality subalgebras.

\subsection{Generalised evaluators and locality characters}
In his seminal work~\mcite{Sp1,Sp2,Sp3,Sp4} on analytic renormalisation, Speer gives an axiomatic formulation for the evaluation of regularised quantities, called generalised evaluators
which he applies to spaces of meromorphic germs  in several variables  singled out by the   regularisation step in his study of Feynman integrals. In the final  renormalisation step, he proposes a generalised evaluator defined by averaging over iterated  one dimensional evaluators successively  applied in each variable.

Speer's pioneering multivariable approach nevertheless lacks a covariance  property since it is coordinate dependent. To circumvent this coordinate dependence, we require the evaluator to be multiplicative on products of germs depending on variables which span mutually perpendicular spaces (the dependence subspaces mentioned above) instead of them having disjoint sets of variables as in Speer's work. This paper  provides a covariant counterpart of Speer's  approach in a sound  mathematical framework, with the aim of setting up a general framework to tackle divergences  in various contexts and with the following three goals in mind.

\begin{question} \mlabel{it:sp1} Equip the space of germs arising  in Speer's  and other multiparameter regularisations with appropriate {\bf locality   polynomial algebra} structures;
\end{question}

\begin{question} \mlabel{it:sp2} Interpret Speer's generalised evaluators  as {\bf locality characters}  on the corresponding (locality) algebras, leading to a general concept of {\bf locality generalised evaluator};
\end{question}

\begin{question} \mlabel{it:sp3} Build a  {\bf transformation group} which  relates different locality evaluators.
\end{question}

\subsection {Outline of the paper}

In the coalgebraic approach to renormalisation in one variable \`a la Connes and Kreimer,  one    calls upon an inductive procedure to deal with mutual compensations of divergences among different levels. Instead,  here we want  to avoid the occurence of such compensations  by means of a multiparameter regularisation, which enables us to  regularise each subdivergence in an autonomous way by introducing a different parameter at each level as we go deeper in the subdivergences. This is possible using locality structures, which take care of keeping the different levels separate in requiring the regularisation map to be a locality algebra homomorphism with range in $\calm _\Q$.

We first provide some background in Section~\mref{sec:mero} on the space $\calm _\Q$ of meromorphic germs with linear poles and rational coefficients on the filtered lattice space $(\R ^\infty, \Z ^\infty)$. For a given inner product $Q$ on the underlying space,  a complement of the holomorphic germs $\calm_{\Q+}$ is given by the subspace $\calm^Q_{\Q-}$ of polar germs defined in Eq.~\eqref{eq:polargerm}.
This gives rise to Laurent expansions and various decompositions and invariants (residues) in $\calm _\Q$ (Theorem~\mref{thm:SpaceDec}), which serve as  the  building blocks of our further study.
We further give a detailed study of dependence spaces of meromorphic germs and their decompositions in \S \mref {ss:DepSp} (Theorem \mref{thm:poldep}), in order to define the orthogonality of meromorphic germs.

With the above orthogonality of meromorphic germs at hand, we carry out a careful study of locality algebras in Section~\mref{sec:local}, focusing on subalgebras of $\calm _\Q$. \S \ref {subsec:localgerm} gives a description of the structures of a locality subalgebra $\calm _{\Q+}^Q(\subS )$ of $\calm _\Q$, which contains $\calm_{\Q+}$ and is generated by a set of  fractions $\cals$. The automorphisms of   such a locality subalgebra, which fix the holomorphic germs and preserve residue type invariants of polar germs, are shown to form a group in \S \ref{sussec:AutoLoc} and \S \mref {subsec:LocGal}. Inspired  by Cartier's cosmic Galois group~\mcite{B,C,CM}, we call it the {\bf locality Galois group} of the locality subalgebra (Definition~\mref{defn:Galoisgroup}).
A reduction theorem (Theorem~\mref{thm:phiiso}) of locality Galois groups is obtained, which shows  that there is a subgroup of the locality Galois group that can be described by special automorphisms of the locality subalgebra generated {\it over $\Q$}  by the same set of fractions $\cals$.

To obtain the structure of locality Galois groups, in Section~\mref{sec:locpoly} we first give a locality variant of polynomial algebras (Definition \ref{defn:localg}) and  then extend the polynomial generation of shuffle product algebras by Lyndon words to the locality setting (Theorem \mref{thm:lynloc}). We finally show (Theorem \mref{thm:conefrac}) that the locality subalgebras generated by certain classes of fractions are locality polynomial algebras. These include Chen fractions arising in multiple zeta functions, described in Example \ref{ex:chen} and Speer fractions described in Example \ref{ex:uexam}, named after Speer in acknowledgment of his work on analytic renormalisation~\mcite{Sp2,Sp3,Sp4} (see also~\mcite{BR,DZ}).

Finally in Section~\mref{sec:app}, we apply the developed results to revisit Speer's approach in the locality framework. We first show (Proposition \ref{prop:FeynmanamplSpeer}) that the space spanned by the fractions arising from  Speer's s-families is precisely the space of the aforementioned Speer fractions, and hence it is a locality polynomial algebra. This addresses Problem~\mref{it:sp1}.

As an analog  of Speer's generalised evaluators, in \S\mref{ss:ge} we introduce the notion of \qorth generalised evaluators  (Definition \mref{defn:geneval}) on a locality subalgebra of $\calm_{\QQ}$. They are linear forms $\cale$ extending the usual evaluation ${\rm ev}_0$ at $0$ defined  on holomorphic germs and, in accordance with the locality principle, they are required to obey the following locality multiplicativity:
\[  f_1\perp^Q f_2\Rightarrow \cale(f_1\, f_2)=\cale(f_1)\,  \cale(  f_2).\]
A proptotype is  the  minimal subtraction evaluator
$\mseva \coloneqq {\rm ev}_0\circ  \pi_+^Q$, see
Eq.~\eqref{eq:mseva}.
\\ \noindent
We further show in Proposition \ref{pp:comeva}, that our locality generalised evaluators   satisfy  the conditions required by Speer for  generalised evaluators  modulo a topological requirement (which lies out of the scope of this paper, and is discussed in  \cite{DPS}) and when  adapting the locality relation appropriately. This addresses Problem~\mref{it:sp2}. Along the way, we  discuss the difference between Speer's coordinate dependent generalised evaluator $\cale_\bullet^{\rm iter}$ and our covariant minimal subtraction scheme $\cale^Q_{\rm MS}$.

The locality Galois group naturally acts on  locality generalised evaluators by  composition (Eq.~\eqref{eq:actionGaloisgroup}). When the locality subalgebra is generated by a locality polynomial algebra  of fractions, the action is shown to be transitive (Theorem~\mref{thm:evaluators}). Consequently, every locality generalised evaluator on such a locality polynomial algebra factors through the minimal subtraction evaluator $\cale^Q_{\rm MS}$.
Applying these results to Chen fractions and Speer fractions, we  obtain Corollary \ref{coro:PolyA},  addressing Problem~\mref{it:sp3}.
Finally we show that multiple zeta values naturally give rise to a locality generalised evaluator $\cale^\zeta$ on  the locality polynomial algebra  $\cala^{\rm Chen}:=\calm_{\Q+}^Q(\Pi ^Q(\mti{\calf}^{\rm Ch}))$  of meromorphic germs at zero with Chen type poles (Eq.~\eqref{eq:mzvchar}). Comparing $\cale^\zeta$ with the minimal subtraction evaluator $\cale^Q_{\rm MS}$ gives a natural element in the locality Galois group.

\section {Meromorphic germs with linear poles}
\mlabel{sec:mero}
This section first summarises definitions and results of \mcite{GPZ3} on the space of meromorphic germs with linear poles, and then studies the dependence space of  meromorphic germs and their decompositions.

\subsection{Spaces of meromorphic germs with linear poles}
We work in a {\bf filtered lattice Euclidean space} $(\R ^\infty, \Z ^\infty, Q)$, consisting of
\begin{enumerate}
	\item a {\bf filtered lattice space} $(\R^\infty,\Z^\infty)$ defined by direct limits
$$\R ^\infty \coloneqq \varinjlim \R ^k,\quad  \Z^\infty \coloneqq \varinjlim \Z ^k,
$$
under the standard embeddings $i_k:\R ^k\to \R ^{k+1}$,
\item an {\bf inner product} $Q$ on the filtered lattice space, defined by a family $Q=(Q_k)_{k\geq 1}$ of inner products
$$  Q_k: \R^k\ot \R ^k \to \RR,$$
such that
$$Q_{k+1}|_{\R ^k\times \R ^k}=Q_k, \quad Q_k (\Z ^k\otimes \Z ^k)\subset \Q.$$
\end{enumerate}

For a field $\K$ with $\Q\subset \K\subset \R$, we denote by $\call _\K(\C ^k)=\call _\K (\K ^k\otimes \bbC)$   the space of linear forms on $\bbC ^k$ which take $\K$-values on $\K ^k$.

On the filtered lattice space $(\R ^\infty, \Z^\infty)$, a meromorphic germ $f$ at zero on $\R^k\otimes \C$ is said to be {\bf $\K$-holomorphic} if it is a holomorphic germ at zero whose power series expansion for any dual basis of $\Z ^k$ has coefficients in $\K$, and to
have {\bf  $\K$-linear poles} if there are vectors $L_1, \ldots, L_k\in (\Z^k)^*\ot \K$ (possibly with repetitions) such that
$$f\,\Pi_{i=1}^k L_i$$
is a $\K$-holomorphic germ.

\begin{rk}
	For the rest of the paper, all meromorphic germs are taken to be at zero unless otherwise stated.
\end{rk}

Let $\calm _\K (\bbC ^k)=\calm _\K (\bbR ^k\otimes \bbC)$ (resp. $\calm _{\K +}(\bbC ^k)=\calm _{\K +}(\bbR ^k\otimes \bbC)$) denote the space of meromorphic (resp. $\K$-holomorphic)  germs with $\K$-linear poles.

The inner product $Q$ induces a family of linear bijections
$$Q_k: \K ^k \to (\K ^k)^*, \quad
	u\mapsto Q_k(u, \cdot)
$$
and
$ Q_k^{-1}: (\K^k)^*\to \K^k.$
This gives rise to maps
$$ p_k\coloneqq Q_k^{-1}i_k^*Q_{k+1}:\C^{k+1} \to \C^k.$$
Explicitly, let $W_k$ be the orthogonal complement of $\C^k$ in $\C^{k+1}$. Then the map $p_k$ is the projection of $\C^{k+1}$ to $\C^k$ along $W_k$ and give rise to the directed systems:
$$p_k: \calm _\K (\bbC ^k)\to \calm _\K (\bbC ^{k+1}), \quad  p_k: \calm_{\K +}  (\bbC ^k)\to \calm_{\K +} (\bbC ^{k+1}),
$$
and the direct limits
$$\calm _\K: =\calm _\K (\bbC ^\infty):=\varinjlim\calm _\K (\bbC ^k), \quad \calm _{\K +}:=\calm _{\K +}(\bbC ^\infty):=\varinjlim \calm_{\K +} (\bbC ^k)
$$  of spaces of meromorphic germs with $\K$-linear poles (resp. $\K$-holomorphic germs).

By restriction, we also let
\vspace{-.2cm}	
\begin{equation}\label{eq:Klinearforms}
	\call _\K\coloneqq \call_\K (\C ^\infty):=\varinjlim\call _\K (\C ^k)
\end{equation}
be the direct limit of spaces of $\K$-linear forms. Notice that $Q$ induces an inner product in
$\call _\K (\C ^\infty)$ which we still denote by $Q$.

As in \mcite{GPZ3},   on a filtered lattice Euclidean space $(\R ^\infty, \Z ^\infty, Q)$, we define a {\bf polar germ} in $\bbC ^k$ with $\K$-coefficients to be a germ of meromorphic functions  of the form
\vspace{-.2cm}	
\begin{equation}\frac {h(\ell_1, \ldots , \ell_m)}{L_1^{s_1}\cdots L_n^{s_n}},
\mlabel{eq:polargerm}
\vspace{-.3cm}	
\end{equation}
where
\vspace{-.1cm}	
\begin{enumerate}
\item
$h$ lies in $\calm_{\K+}(\bbC^m)$,
\item
$\ell_1, \ldots, \ell_m, L_1, \ldots ,L_n$  lie in $\call _\K (\C ^k) $, with $L_1, \ldots ,L_n$ linearly independent, such that
$$Q(\ell_i, L_j)=0 \quad \forall i\in \{1, \ldots, m\}, j\in \{1, \ldots, n\},
$$
\item
$m$ is a nonnegative integer and $ n, s_1,\ldots, s_n$
are positive integers.
\end{enumerate}
The convex cone (resp. subspace space) spanned by $L_1,\ldots,L_n$ is called the {\bf supporting cone} (resp. {\bf supporting space}) of the polar germ $f$.
The supporting space, independent of the presentation of the germ in the form of the fraction (see~\cite[Lemma 2.9]{GPZ3}), is denoted by $\supsp(f)$. Recall from \cite[Definition 5.1]{GPZ3} that, for a polar germ $\frac {h(\ell_1, \ldots , \ell_m)}{L_1^{s_1}\cdots L_n^{s_n}}$, the integer
\vspace{-.2cm}	
$$\text{p-ord} \left(\frac {h(\ell_1, \ldots , \ell_m)}{L_1^{s_1}\cdots L_n^{s_n}}\right)\coloneqq |(s_1,\ldots,s_n)|\coloneqq s_1+\cdots +s_n
\vspace{-.1cm}	
$$
is well defined, called the {\bf p-order} of the polar germ.

Let $\calm ^Q_{\K-}(\C ^k)$ denote the linear space spanned by polar germs on $\C ^k$ with $\K$-linear poles. Again we have a directed system
\vspace{-.2cm}	
$$p_k: \calm^Q_{\K-} (\bbC ^k)\to \calm^Q_{\K-} (\bbC ^{k+1})
\vspace{-.2cm}	
$$
and the direct limit
\vspace{-.2cm}	
$$\calm ^Q_{\K-}\coloneqq \calm ^Q_{\K-}(\bbC ^\infty)\coloneqq \varinjlim \calm^Q_{\K-} (\bbC ^k)
\subseteq \calm_\K (\bbC^\infty)
\vspace{-.1cm}	
$$
of polar germs with $\K$-linear poles.

Polar germs split according to their supporting subspaces and p-orders as follows.
\vspace{-.1cm}	
\begin{lemma} \cite[{Lemma 3.5}]{GPZ3}
Suppose
$\sum\limits _{i=1}^r S_i=0
$
for a sum of holomorphic germs and $\K$-polar germs.
For any  linear $\K$-subspace $W$ of $V$ and $N \in \ZZ _{>0}$, we have
$ \sum\limits_i \!^\prime S_i =0, $
where the sum is over the terms $S_i$ with $\supsp(S_i)=W$ and $ \pord (S_i)=N$, with the convention that the sum over an empty set is zero.
\mlabel{lem:SepSpace}
\end{lemma}

Following \cite[Definition 3.2]{GPZ3},
\begin{enumerate}
\item a family of (convex) cones is called {\bf properly positioned} if every pair of cones in the family intersect along their faces, including the zero dimensional face at $0$, and their union does not contain a straight line;
\item a family of polar germs {\bf properly positioned} if, for each of the polar germs, there is a choice of a supporting cone such that the resulting family of cones is properly positioned;
\item a family of polar germs is called {\bf  projectively} properly positioned if it is properly positioned and none of the denominators of the polar germs is proportional to another.
\end{enumerate}

Here is a useful criterion for the linear independence of polar germs.

 \begin {prop} \cite[Proposition 3.6]{GPZ3} A finite family of polar germs with projectively properly positioned supporting cones is linearly independent.
 \mlabel{pp:proper}
 \end {prop}

\subsection{Laurent expansions and the induced decompositions}
There are several decompositions of meromorphic germs with linear poles. Recall that a convex cone is called {\bf simplicial} if it is spanned by a set of linearly independent vectors.
\begin {prop} \mlabel{pp:lau}
\cite[Theorem 4.13.]{GPZ3} For any  $f$  in $\calm _\K $, there exist a properly positioned family $\conefamilyc$ of simplicial cones together with a family of $\K$-polar germs   $\{S_j\}_{j\in J}$ supported on $\conefamilyc$ $($in the sense that a supporting cone of each $S_j$ is in $\conefamilyc)$,  and a holomorphic germ $h$, such that
\begin{equation}\mlabel{eq:decgerm}
	f= \sum _{j\in J} S_j+ h.
\end{equation}
\end{prop}
Eq.~\meqref{eq:decgerm} is called a {\bf Laurent expansion} of $f$ supported on $\conefamilyc$ and it  is unique up to subdivisions of the properly positioned family of simplicial cones.

For $p\in \Z_{\geq 0}$, $d\in \Z_{\geq 0}$  and  a finite dimensional $\K $-subspace $U\subset \R ^\infty$, let
\begin{itemize}
	\item $\calm_{\K}^p$ denote the linear span of $\K$-polar germs of p-order $p$;
	\item   $\calm_{\K,d}$ denote the linear span of $\K$-polar germs whose supporting cones have dimension $d$;
\item  $\calm_{\K ,U}$ denote the linear span of $\K$-polar germs with supporting space $U$.
\end{itemize}

\begin{theorem}\mlabel{thm:SpaceDec}
	\cite[Theorem~5.3]{GPZ3} We have the decompositions
{\small
\begin{eqnarray}
{\calm} _{\K}&=&\bigoplus_{p\geq 0 } {\calm} _{\K}^p,\mlabel{eq:pgrad}\\
{\calm} _{\K }&=&\bigoplus_{d\geq 0} {\calm} _{\K, d},\mlabel{eq:dgrad}\\
{\calm} _{\K }&=&\bigoplus_{U \subset \R^\infty } {\calm} _{\K, U}.\mlabel{eq:ugrad}
\end{eqnarray}
}
In particular, there is a decomposition $($see also \mcite{BV1}$)$:
\[\calm_{\K}= \calm_{\K+}\oplus \calm ^Q_{\K-}.\]
\end{theorem}

We give further notations.
\begin{enumerate}
	\item Corresponding to  the decomposition in Eq.~\meqref{eq:pgrad}, let $q$ be the highest p-order of the polar germs in a (thus every) Laurent expansion of $f$. Define the
	{\bf p-residue}  of $ f$~\cite[Definition 6.1]{GPZ3} by
	\begin{equation}
	{\rm  pRes} (f)\, {:}=\sum_{\pord(S_i)=q}  \frac {h_i(0)}{\vec L_i^{\vec s_i}}.
	\mlabel{eq:plt}
	\end{equation}
	\item Corresponding to  the decomposition in Eq.~\meqref{eq:dgrad}, let $e$  be the largest   among the dimensions of the supporting spaces of the polar germs in a (thus every) Laurent expansion of $f$. Define the
	{\bf d-residue} of $ f$ by
	\begin{equation}
	{\rm  dRes} (f)\, {:}=\sum_{{\rm dim} (\supsp(S_i))=e}  \frac {h_i (0)}{\vec L_i^{\vec s_i}}.
	\label{eq:dlt}
	\end{equation}
\end{enumerate}

\begin{rk} As proved in  \cite[Proposition 6.2]{GPZ3}, the p-residue of a meromorphic germ with linear poles depends neither on the choice of a Laurent expansion nor on the choice of the inner product. The d-residue does not depend on the choice of a Laurent expansion, but it does depend on the choice of the inner product.
\end{rk}

\subsection {Dependence subspaces}
\mlabel{ss:DepSp}
	
For any subset $U$ of $\calm_\Q$, let
$\Q U$ denote the $\Q$-subspace of $\calm_\K$ spanned by $U$.

A {\bf simplex fraction} is a fraction of the form $ \frac {1}{L_1^{s_1}\cdots L_k^{s_k}}$, where $L_1, \ldots, L_k \in \call _\Q $ are linearly independent and $s_i\in \Z _{>0}$, $i=1, \ldots, k$.
Let $\calf $ be the set of all simplex fractions over $\Q$.
Then for any inner product $Q$ in $(\R ^\infty, \Z ^\infty)$, we trivially have $\Q \calf \subset \calm _{\Q-}^Q.$

\begin{ex}
\mlabel{ex:chen}
In the Euclidean filtered lattice space $(\R ^\infty, \Z ^\infty, Q)$, let $\mathcal{B}\coloneqq (e_i)_{i\in \Z_{>0}}$ be an orthonormal basis.
Let $z_i$ be the coordinate function corresponding to $e_i$.  A fraction of the form
\begin{equation}\label{eq:chenfrac}\frakf\pfpair{s_1,\ldots,s_k}{{u_1},\ldots,{u_k}}\coloneqq \frac 1{z_{u_1}^{s_1}(z_{u_1}+z_{u_2})^{s_2}\cdots (z_{u_1}+z_{u_2}+\cdots+z_{u_k})^{s_k}}\,, \   u_i, s_i\in \Z_{>0},k\in \N, u_i\neq u_j  \, {\rm if}\ i\neq j,
\end{equation}
 is called a {\bf Chen fraction}.
The set of Chen fractions is denoted by
$$\calf^{\rm Ch}\coloneqq \calf^{{\rm Ch}, Q,\mathcal{B}}.$$
\end{ex}

We borrow the following definitions from \mcite{GPZ3}.
 A  meromorphic function $f$ on $\C ^k$ of the form $f=g(L_1, \ldots, L_n)$,  where $L_1, \ldots , L_n$ are linear forms on $\C ^k$ and $g$ a meromorphic function on $\C ^n$,  is said to {\bf depend on} the linear subspace of $(\C ^k)^*$ spanned by $L_1, \ldots, L_n$.
 One can check that if $f$ depends on $V_1$ and $V_2$, then it depends on $V_1\cap V_2$. Thus it makes sense to set the following definition.
\begin{defn} The {\bf dependence subspace} ${\rm Dep}(f)$ of $f$ is the smallest linear subspace of $(\C^k)^*$ on which $f$ depends.
	\end{defn}

\begin{ex} Let $(e_1, e_2, \ldots)$ be an orthonormal basis of $(\R ^\infty, \Z ^\infty, Q)$.
$${\rm Dep }\left(\frac 1{z_1(z_1+z_2)}+\frac 1{z_2(z_1+z_2)}-\frac {2}{z_1(z_1+2z_2)}-\frac 1{z_2(z_1+2z_2)}+\frac 1{z_3}\right)=\{e_3\},
 $$
since the sum of the first four terms is zero.
\end{ex}

\begin{defn} \mlabel{defn:depspace}
Two meromorphic germs $f$ and $g$ in $\calm_\Q$ are called {\bf $Q$-orthogonal}, which we write $f\perpq g$, if their dependence subspaces are  orthogonal.
\end{defn}

\begin{ex} \begin{enumerate}
\item 	
 Let $(e_1, e_2, \ldots)$ be an orthonormal basis of $(\R ^\infty, \Z ^\infty, Q)$. We have
\vspace{-.1cm}	
$$\frac{1}{z_1+z_2}  \, \perpq \, (z_1-z_2).
\vspace{-.1cm}	
$$
\item Polar germs are precisely germs of the form
$h/M$
for $h $ in $ \calm_{\Q +}(\bbC^\infty)$ and $M$ given by products of powers of linearly independent linear forms, such that $h \perp^Q M$.
\end{enumerate}
\end{ex}
An element $f\in \calm_\Q$ is of the form $f=\frac{h}{\ell_1 \cdots \ell_r}$ for a holomorphic germ $h$ and linear forms $\ell_1,\ldots,\ell_r$. The next lemma shows that the factors in the fraction can be chosen to have their dependence subspaces contained in the dependence subspace of $f$.
\begin{lemma}
\mlabel{lem:depsupp}
For any $f$ in $ \calm_\Q$, there are linear forms $\ell_i=\ell_i(L_1,\ldots,L_n), i=1,\ldots,p,$ and a holomorphic germ $h=h(L_1,\ldots,L_n)$ for a basis $L_1,\ldots,L_n$ of ${\rm Dep}(f)$, such that
$\displaystyle{f=\frac{h}{\ell_1\cdots\ell_p}}.$
\end{lemma}

\begin{proof}
Let $f$ be in $\calm_\Q(\C^k)$ for some $k\geq 1$.
We extend a basis $L_1, \ldots, L_n$ of ${\rm Dep}(f)$ to a basis  $L_1, \ldots, L_n, \ldots, L_k$ of $(\C ^k)^*$. Since $f$ is in $\calm_\Q(\C^k)$, there are linear combinations $\ell_1, \ldots, \ell _m$ of $L_1, \ldots, L_k$ such that the product
$ \ell _1\cdots \ell_m\, f$ is in $\calm _{\Q+}(\C^k),$
that is,
\begin{equation}\label{eq:lghol}\ell _1\cdots  \ell_m\,f(L_1, \ldots, L_n)\in \calm _{\Q+}(\C^k).
\end{equation}
By rearrangement, we can assume that $\ell_1, \ldots, \ell_p$ are linear combinations of $L_1, \ldots, L_n$ only; while $\ell_{p+1}, \ldots, \ell_m$ have nontrivial linear contributions from the extra linear forms $L_{n+1}, \ldots, L_k$.
\nc{\bell}{\lambda}
Then we can choose a tuple $ (a_{n+1}, \ldots , a_k)\in \C^{k-n}$ such that the maps
$$(L_1, \ldots, L_n)\mapsto \bell_i(L_1, \ldots, L_n)\coloneqq \ell_i(L_1, \ldots, L_n, a_{n+1}, \ldots , a_k),\quad i=p+1, \ldots, m, $$
are affine with $\bell_i(0,\ldots,0)\neq 0$.
Consequently, the maps  $(L_1, \ldots, L_n)\longmapsto \frac 1{\bell_i(L_1, \ldots, L_n)} $ are holomorphic germs.

Thus setting $L_{n+1}=a_{n+1}, \ldots, L_k=a_k$ in Eq.~\meqref{eq:lghol} yields  a holomorphic germ
$$h:(L_1, \ldots, L_n)\longmapsto  \ell_1\cdots \ell_p\,\bell_{p+1}\cdots \bell_m\,  f, $$ from which we define another holomorphic germ
\vspace{-.2cm}	
$$\tilde h(L_1, \ldots, L_n)\coloneqq \frac{h(L_1, \ldots, L_n)}{  \bell_{p+1}\cdots \bell_m}.
\vspace{-.2cm}	
$$
Hence,
\vspace{-.2cm}	
$$f=f(L_1, \ldots, L_n)=\frac{h(L_1,\ldots,L_n)}{\ell_1\cdots \ell_p\bell_{p+1}\cdots \bell_m}=\frac{\tilde h(L_1, \ldots, L_n)}{\ell_1(L_1, \ldots, L_n)\cdots \ell_p(L_1, \ldots, L_n)}
\vspace{-.1cm}	
$$
is of the desired form.
\end{proof}

\begin{theorem}
\mlabel{thm:poldep}
Write a germ $f$ in $ \calm_\Q$ according to the decomposition in Eq.~\eqref{eq:ugrad}$:$
\begin{equation}\label{eq:decf}f=\sum_{U\in \calu} f_U+f_0,\end{equation}
where  $\calu$ is a finite set of nonzero finite-dimensional subspaces of $\RR^\infty$,  $0\neq f_U $ is a sum of polar germs   with  supporting space $U$ and $f_0$ lies in  $ \calm_{\Q +}$.
We have
$${\rm Dep}(f)= \sum_{U\in \calu} {\rm Dep}(f_U)+ {\rm Dep}(f_0).
$$
\end{theorem}
\begin {proof} Since $f=\sum f_U+f_0$, clearly we have
$${\rm Dep}(f)\subset\sum _{U\in \mathcal U} {\rm Dep}(f_U)+{\rm Dep}(f_0).
$$

It remains to show that ${\rm Dep}(f_U)\subset {\rm Dep}(f)$ for all $U\in \mathcal U$  and ${\rm Dep}(f_0)\subset {\rm Dep}(f)$.

By Lemma~\mref{lem:depsupp}, there are linear forms $\ell_1,\ldots,\ell_p$ and a homomorphic germ $h$, all with dependent spaces in ${\rm Dep} (f)$ such that
$$f=\frac{h}{\ell_1\cdots\ell_p}=\frac{h(L_1,\ldots,L_n)}{\ell_1(L_1,\ldots,L_n)\cdots \ell_p(L_1,\ldots,L_n)},$$
for a basis $L_1,\ldots,L_n$ of ${\rm Dep}(f)$.
Then we can take the Laurent expansion of $f$ in ${\rm Dep}(f)$ by~\cite[Theorem~2.11]{GPZ3}.  Thus for all the polar germs in this Laurent expansion of $f$, their linear poles and holomorphic numerators have dependence space in ${\rm Dep}(f)$.
This gives another decomposition
$$ f=g_0+ \sum_{V \subset {\rm Dep}(f)} g_V$$
according to Eq.~\meqref{eq:ugrad}.
Comparing with the decomposition of $f$ in Eq.~(\mref{eq:decf}) as  a sum of a holomorphic germ and polar germs, we have
$$f=f_0+\sum_{U\subset\mathcal U} f_U=g_0 +\sum_{V\subset {\rm Dep}(f) } g_V.
$$
Using the uniqueness of  the decomposition in Eq.~\meqref {eq:ugrad}, we infer  that   for any subspace $U\subset\mathcal U$ (resp. $U=0$), there is a space $V\subset { \rm Dep}(f)$ (resp. $V=0$) such that $f_U=g_V$.
This implies that ${\rm Dep}(f_U)$ is contained in ${\rm Dep}(f)$.
Thus the proof is completed.
\end{proof}

\begin {lemma}\label{lem:DepfU} For a nonzero rational linear combination $f=\sum_{i\in I} \alpha_i S_i\in \calm_\Q(\C^k)$ of simplex fractions $S_i, i\in I,$ with the same supporting space $U$, we have
${\rm Dep}(f)=U.$
\mlabel {lem:fracdep}
\end{lemma}

\begin {proof}  Clearly, ${\rm Dep}(f)\subseteq U.$
Suppose ${\rm Dep}(f)\subsetneq U$.
Theorem~\mref{thm:poldep} gives
$$f=g_0+\sum_{V \subset {\rm Dep}(f)} g_V$$
where  the terms in $g_V$ have supporting space $V$. Hence
$$0=f-\sum_{i\in I}\alpha_i\, S_i=g_0+\sum_{V\subset {\rm Dep}(f) } g_V -\sum_{i\in I}\alpha_i\, S_i.
$$
From ${\rm Dep}(f)\subsetneq U$, we have $V\subsetneq U$ in the above sum. Thus the above sum is the decomposition of $0$ according to the supporting spaces in Eq.~\meqref{eq:ugrad}, in which $\sum_{i\in I} \alpha_i\, S_i$ is the component with supporting space $U$. Thus
$\sum_{i\in I} \alpha_i\, S_i=0$, meaning $f=0$. This is a contradiction.
\end{proof}

This leads to the following statement on sums of polar germs.
\begin{prop} If $f=\sum f_i$ is a nonzero sum of polar germs $f_i$ with the same supporting space $U$, then
$U$ is a subset of $ {\rm Dep}(f).$
\mlabel {prop:polardep}
\end{prop}
\begin{proof}
This follows from Lemma \mref{lem:DepfU} after evaluation of the numerators of the polar germs $f_i$ at appropriate arguments in the spirit of the proof of   \cite[Theorem 3.7]{GPZ3}.  Indeed, let us write the polar germs $ f_i=h_i\,  S_i $ where
\vspace{-.2cm}	
$$S_i=\frac {1}{L_{  1}^{s_1}\cdots L_{n}^{s_{n_i}}}, \quad s_j\in \Z_{\geq 0},
\vspace{-.2cm}	
$$
are simplex fractions with the same supporting space  and $h_i$ are holomorphic germs in some common set of variables $\ell_{n+ 1}, \ldots, \ell_{ k}$  which complete the independent linear forms $L_{ 1}, \ldots, L_{ n}$ arising in the $S_i$'s to an orthonormal basis of $\R^k$. Since $f\neq 0$ we can assume without loss of generality that none of the holomorphic germs $h_i$  is  identically zero. Hence, there is   some tuple  $(\ell_{n+ 1}^0, \ldots, \ell_{ k}^0)$ such that $\alpha_i \coloneqq h_i(\ell_{n+ 1}^0, \ldots, \ell_{k}^0)\neq 0$ for all $i$. We write $f=f(L_1, \ldots, L_n, \ell_{n+1}, \ldots, \ell_k) $ and take the specialisation $g(L_1,\ldots,L_n)\coloneqq f(L_1,\ldots,L_n,\ell_{n+1}^0, \ldots, \ell_k^0)$. Then ${\rm Dep}(f)\supset {\rm Dep}(g)$. Applying Lemma \ref {lem:fracdep} to $g= \sum_i h_i(\ell_{n+1}^0, \ldots, \ell_k^0)\, S_i$ with $\alpha_i=h_i(\ell_{n+1}^0, \ldots, \ell_k^0)$, we obtain ${\rm Dep}(g)={\rm Supp}(g)$. This completes the proof.
\end{proof}
The following result shows that without loss of generality, we can assume that a sum $f$ of polar germs with the same supporting space  can be written as a sum of polar germs whose numerators are holomorphic germs with  dependence space in ${\rm Dep}(f)$.

\begin {prop}
\mlabel {pp:RewPG}
Let $ f=\sum  h_i\,S_i $ be a nonzero sum of polar germs  with the same supporting space, where $S_i$ is a simplex fraction and $h_i$ is a holomorphic germ. Then $f$ can be written as a sum $f=\sum \tilde h_i\,S_i$ of polar germs where the  $\tilde h_i$'s are now  holomorphic germs with dependence spaces in ${\rm Dep}(f)$.
\end{prop}
\begin {proof}By  Proposition \mref{prop:polardep}, the common supporting space  $U$ lies in $ {\rm Dep}(f)$.  Let $L_1, \ldots , L_n$ be a basis of $U$ which we extend to a basis $L_1, \ldots, L_n, \ell _{n+1}, \ldots, \ell _{m}$ of ${\rm Dep }(f)$ and then further to a basis $L_1, \ldots, L_n, \ell _{n+1}, \ldots, \ell _k$ of $(\C ^k)^*$ with $Q(L_i, \ell _j)=0, 1\leq i\leq n, n+1\leq j\leq k$. Thus, $S_i$ is a simplex fraction in the variables $L_1, \ldots, L_n$ and
$$f=f(L_1,\ldots,L_n,\ell_{n+1},\ldots,\ell_k)=\sum h_i(\ell_{n+1}, \ldots , \ell _{k})\, S_i(L_1, \ldots, L_n).
$$
By the definition of dependence space, $f$ does not depends on $\ell _{m+1}, \ldots, \ell _{k}$, so
$$f=f(L_1,\ldots,L_n,\ell_{n+1},\ldots,\ell_m,0,\ldots,0)=\sum h_i(\ell_{n+1}, \ldots , \ell _{m}, 0, \ldots, 0)\, S_i(L_1, \ldots, L_k)
$$
as announced.
\end{proof}

\section{Locality transformation groups on meromorphic germs}
\mlabel{sec:local}
In this section, we study meromorphic germs with linear poles in the context of locality algebras. We then introduce the locality Galois group defined as a group of automorphisms of meromorphic germs in this locality framework.

\subsection{Locality algebras of meromorphic germs}
\label{subsec:localgerm}
We give general background on locality algebras and then focus on locality subalgebras of meromorphic germs.
\subsubsection{Locality algebras}\label{subsec:localg}
We recall notations on locality structures from~\mcite{CGPZ1}.

\begin{defn}
	A {\bf \loc set} is a  couple $(X, \top)$ where $X$ is a set and
	$$ \top\coloneqq X\times_\top X \subseteq X\times X$$
	is a binary {\em symmetric} relation, called a {\bf locality relation}, on $X$. For $x_1, x_2\in X$,
	denote $x_1\top x_2$ if $(x_1,x_2)\in \top$.
\end{defn}

For a subset $U\subset X$, the {\bf  {polar} subset} of $U$ is
\begin{equation*}
U^\top\coloneqq \{x\in X\,|\, (x,U)\subseteq \top \}.  			
\end{equation*}

For locality sets $(X,\top_X)$ and $(Y,\top_Y)$, a map $f:X\to Y$ is called a {\bf locality map} if
\begin{equation} \mlabel{eq:locmap}
x_1\top_X x_2 \Longrightarrow f(x_1)\top_Y f(x_2), \quad \forall x_1, x_2\in X.
\end{equation}

We give some examples that will be further explored in the sequel.
\begin{ex}
	\begin{enumerate}
\item For any nonempty set $X$, being distinct: $x_1\top x_2$ if $x_1\neq x_2,$ defines a locality relation on $X$;
		\mlabel{it:locex1}
\item The $Q$-orthogonality relation $\perp^Q\subset  \calm_\QQ\times \calm_\QQ$   of Definition \ref{defn:depspace}, turns  $\calm_\QQ$ into a locality set.
\mlabel{it:locex2}
\item Let $(X, \top)=(\ZZ_{>0},\top)$ be the locality set in \meqref{it:locex1} and $(\calm_\QQ,\perp^Q)$ the locality set in \meqref{it:locex2}. With the notation in Example~\mref{ex:chen}, the map
$$ f: X\to \calm_\QQ, \quad
n\mapsto \frakf\pfpair{1}{n}\coloneqq \frac{1}{z_n},\quad  n>1,$$
is a locality map.
\mlabel{it:locex3}
\end{enumerate}
\end{ex}

Other algebraic structures can be generalised to the locality setting.
\begin{defn} \mlabel{defn:lsg}
	\begin{enumerate}
		\item
A {\bf \loc vector space} is a vector space $V$ equipped with a \loc relation $\top$ which is compatible with the linear structure on $V$ in the sense that, for any  subset $X$ of $V$, $X^\top$ is a linear subspace of $V$.
\item A (nonunitary) {\bf \loc  algebra} over $K$ is a \loc vector space $(A,\top)$ over $K$ together with a map
$$ m_A: A\times_\top A \to A,
(u,v)\mapsto  u\cdot v=m_A(x,y) \quad \text{for all } (u,v)\in A\times_\top A$$
satisfying the following variations of the associativity and distributivity.
\begin{enumerate}
	\item[(a)] For $u, v, w\in A$ with $u\top v, u\top w, v\top w$, we have
		\begin{equation}
(u\cdot v) \top w,\quad  u\top (v\cdot w), \quad (u\cdot v) \cdot w = u\cdot (v\cdot w). 	
	\mlabel{eq:asso}
\end{equation}
\item[(b)] For $u, v, w\in A$ with $u\top w, v\top w$ (and hence $(u+v)\top w, w\top (u+v)$), we have
$$ (u+v)\cdot w = u\cdot w + v\cdot w, \quad   w\cdot (u+v) = w\cdot u+w\cdot v, $$
$$ (ku)\cdot w =k(u\cdot w),\quad  u\cdot (kw)=k(u\cdot w), \ k\in K.$$
\end{enumerate}
\item A {\bf unitary \loc algebra} is a \loc algebra $(A,\top, m_A)$ with a unit $1_A$ such that, for each $u\in A$, we have  $1_A\top u$ and
$$1_A\cdot u=u\cdot 1_A=u.$$
We shall omit explicitly  mentioning the unit $1_A$ unless doing so generates ambiguity.
\item Let $(A,\top_A)$ be a locality algebra. A subspace $B$ of $A$ is called a {\bf (resp. unitary) locality subalgebra of $A$} if, with the restricted relation
\vspace{-.2cm}	
$$\top_B\coloneqq \top_A \cap (B\times B)
\vspace{-.2cm}	
$$
of $\top_A$ to $B$, the pair $(B,\top_B)$ is a (resp. unitary) locality algebra.
\item Let $(A,\top_A)$ be a commutative locality algebra and $C$ a unitary locality subalgebra of $A$. A subspace $B$ of $A$ is called a {\bf (resp. unitary) locality $C$-subalgebra of $A$} if $B$ is a (resp. unitary) locality subalgebra of $A$ that contains $C$.
\end{enumerate}
\end{defn}

Given two locality algebras $(A_i, \top_i), i=1,2$,
a (resp. {\bf unitary}) {\bf locality algebra homomorphism} is a linear map   $\varphi:A_1\longrightarrow A_2$  such that
$a  \top_1 b$ implies $\varphi(a)\top_2\varphi(b)$ and
	$\varphi(a \cdot b)=\varphi(a )\cdot \varphi(b)$
(resp. and $\varphi(1_{A_1})=1_{A_2}$).

\begin{ex}\mlabel{ex:mero}
\cite[Corollary 3.23]{CGPZ1} With the relation $\perp^Q$ of Definition \ref{defn:depspace}, the pair $(\calm_\Q, \perp^Q)$ is a unitary locality algebra and
	the projection
\vspace{-.1cm}	
	\begin{equation}\label{eq:piqplus}\pi_+^Q: \calm_{\Q}=\calm_{\Q+}\oplus \calm^Q_{\Q-} \to  \calm _{\Q+}
\vspace{-.1cm}	
	\end{equation}
along $\calm ^Q_{\Q-}$ is a unitary \qorth algebra homomorphism.
\end{ex}

\begin {defn} \mlabel{de:auto} For a unitary locality algebra $(A, \top)$, a unitary locality endomorphism  of $A$ is a {\bf locality  automorphism} if it is invertible,  preserves the unit, and the inverse map is a  locality  algebra homomorphism.
Let  ${\rm Aut}^\top(A)$ denote the set of   locality automorphisms of $A$.
\end{defn}

We note that  ${\rm Aut}^\top(A)$ forms a group for the composition, called the {\bf locality automorphism group} of $A$.
There are counter examples that a bijective \qorth homomorphism needs not be a \qorth automorphism.

\subsubsection{Locality subalgebras of meromorphic germs}
In the sequel, we consider locality subalgebras $(\cala, \perp^Q)$  of the locality algebra $(\calm_\Q,\perp^Q)$ in Example~\mref{ex:mero}. For a subset $U$ of $\calm_\Q$, let
\begin{equation} \mlabel{eq:unit}
\mti{U}\coloneqq U\cup \{1\},
\vspace{-.2cm}	
\end{equation}
with $1$ being the constant function.

We first   give the structure of locality subalgebras of $\calm_{\Q}$ generated by a set, with rational coefficients or $\calm _{\Q+}$ coefficients. As we shall see, a careful analysis using tools such as supporting and dependent spaces is needed when extending the notion of subalgebra    to the locality setting.

Given a subset $\calu$ of $\calm_\Q$, let
\vspace{-.2cm}	
$$\Pi^Q (\calu)\coloneqq \bigg\{ \prod _i s_{i} \, \bigg|\, s_{i} \in \calu, \forall i, s_{i}\perpq  s_{j},\, \forall i\not =j \ \bigg\}
\vspace{-.2cm}	
$$
be the set of meromorphic germs \qorth generated by $\calu$. With the notation of Eq.~\meqref{eq:unit},  we have
\vspace{-.2cm}	
$$ \Pi^Q(\mti{\calu})=\Pi^Q(\calu) \cup \{1\}.
\vspace{-.2cm}	
$$

\begin{prop}\mlabel{prop:genalg} Given a set $\cals$ of simplex fractions,   the subspace of $\Q \mti{\calf}$
\vspace{-.2cm}	
  \[ \QsubS\coloneqq \bigg\{ \sum_{i} c_i\, S_i  \, \bigg|\, c_i\in \Q, S_i\in \Pi^Q(\mti{\cals}) \bigg\}
  \vspace{-.2cm}	
  \]
spanned by $\Pi^Q(\mti{\cals})$ is a unitary \qorth subalgebra  of $\Q \calf $.
\end{prop}
Thus $\QsubS$ is the unitary locality subalgebra of $\Q\mti{\calf}$ generated by $\cals$.
\begin {proof} Since $c_0=1$ serves as the unit, we just need  to prove that for $f, g$ in $\Pi^Q (\cals)$ with $f\perpq  g$,  $fg$ lies in $ \Q \Pi^Q( \cals)$. According to the grading in Eq.~\meqref{eq:ugrad}, we write
\vspace{-.1cm}	
$$f=\sum _U f_U, \quad  g=\sum_V g _V,
\vspace{-.1cm}	
$$
where $f_U$ is the sum of simplex fractions with the same supporting  space $U$ and $g_V$ is the sum of simplex fractions with the same supporting space $V$. By Theorem~\ref{thm:poldep}, we have
\vspace{-.1cm}	
$${\rm Dep}(f)=\sum_U {\rm Dep}(f_U), \quad {\rm Dep}(g)=\sum_V {\rm Dep}(g_V)
\vspace{-.1cm}	
$$
from which we infer that, for any $U$ and $V$ appearing in the decompositions of $f$ and $g$,  \[f\perpq  g\Longrightarrow f_U\perpq g_V.\]
By Lemma \mref{lem:fracdep}, each $f_U \not =0$   (resp. $g_V\not =0$),  being a sum of simple fractions with the same supporting space $U$ (resp. $V$), gives $U={\rm Dep}(f_U)$ (resp. $V={\rm Dep}(g_V)$). Thus we have
$$U\perpq  V. $$
Hence the products $f_Ug_V$   arising in the decomposition $f\, g=\sum_{U, V}f_Ug_V$ lie  in $\Q\Pi^Q(\cals)$. Consequently the product $f\, g$ also lies  in $\Q\Pi^Q(\cals)$.
\end{proof}

\begin{prop}\mlabel{prop:genalg2} Given a set $\cals$ of simplex fractions, the set
\vspace{-.1cm}	
  \[\calm _{\Q+}^Q\big(\subS \big)\coloneqq \bigg\{ \sum_i h_i S_i,\bigg|\, h_i\in \calm _{\Q +}, S_i\in \subS,   h_i\perpq S_i \bigg\}
  \vspace{-.1cm}	
  \]
is a unitary \qorth subalgebra of $\calm _\Q$.
\end{prop}
Thus $\calm _{\Q+}^Q\big(\subS\big)$ is the unitary locality $\calm_{\Q+}$-subalgebra of $\calm_\Q$ generated by  $\cals$.

\begin {proof} Clearly,  $\calm _{\Q+}^Q(\subS)$ is a vector space. As in the proof of Proposition~\mref{prop:genalg}, we only need to verify that, for $a, b \in \calm _{\Q+}^Q(\subS )$ with $a\perpq  b$, the locality product $ab$ is well defined and lies in $\calm _{\Q+}^Q(\subS )$. To complete this, in accordance with the grading in Eq.~\meqref{eq:ugrad}, write
\vspace{-.1cm}	
$$a=\sum_U \sum _i a_{U_i}S_{U_i}=\sum _Ua_U, \quad b=\sum_V \sum _j b_{V_j}T_{V_j}=\sum _V b_V,
\vspace{-.1cm}	
$$
where $a_{U_i}, b_{V_j}\in \calm _{\Q+}$, $S_{U_i}, T_{V_j}\in \Pi ^Q(\cals)$ and $a_{U_i}\perpq  S_{U_i}$, $b_{V_j}\perpq  T_{V_j}$, ${\rm Supp }(S_{U_i})=U$, ${\rm Supp}(T_{V_j})=V$ (with the convention that $\supsp(h)=0$ for $h\in \calm_{\Q+}$).

Theorem~\mref{thm:poldep} gives ${\rm Dep }(a_U)\subset {\rm Dep}(a)$ and ${\rm Dep}(b_V)\subset {\rm Dep}(b)$. Since $a\perpq  b$ means ${\rm Dep}(a)\perpq  {\rm Dep}(b)$, we have ${\rm Dep}(a_U)\perpq  {\rm Dep}(b_V)$, that is $a_U\perpq  b_V$.

Now let $a_U\not =0$ and $b_V\not =0$.
\begin{enumerate}
  \item By Proposition~\mref{prop:polardep}, we have $U\subset {\rm Dep} (a)$ and $V\subset {\rm Dep}(b)$. So $U\perpq  V$ and $S_{U_i}\perpq  T_{V_j}$;
  \item From ${\rm Dep }(a_U)\subset {\rm Dep}(a)$ and $V\subset {\rm Dep}(b)$, we obtain $a_U \perpq  T_{V_j}$. Similarly, $b_V \perpq  S_{U_i}$;
  \item By Proposition~\mref{pp:RewPG}, there exist holomorphic germs $\tilde a_{U_i}$ and $\tilde b_{V_j}$ with
\vspace{-.2cm}	
  $${\rm Dep}(\tilde{a}_{U_i})\subset {\rm Dep}(a_U), \ \tilde a_{U_i} \perpq  S_{U_i},\ {\rm Dep}(\tilde{b}_{V_j})\subset {\rm Dep}(b_U),\ \tilde b_{V_j} \perpq  T_{V_j},
  \vspace{-.2cm}	
  $$
  such that
\vspace{-.2cm}	
  $$ a_U=\sum _i\tilde a_{U_i}S_{U_i},\quad b_V=\sum_j \tilde b_{V_j}T_{V_j}.
\vspace{-.2cm}	
  $$
From $a_U\perpq  T_{V_j}$ and $\depsp(\tilde{a}_{U_i})\subset \depsp(a_U)$, we have $\tilde a_{U_i} \perpq  T_{V_j}$. Likewise,  $\tilde b_{V_j} \perpq  S_{U_i}$.
\end{enumerate}
In summary, $\tilde a_{U_i}, \tilde b_{V_j}, S_{U_i}, T_{V_j}$ are mutually $Q$-orthogonal. Thus
$a_Ub_V=\sum\limits _{i,j} \tilde a_{U_i}\tilde b_{V_j}S_{U_i}T_{V_j},
$
is defined and gives an element in $\calm _{\Q+}^Q(\subS) )$. Therefore,
$ab=\sum\limits_{U,V}a_Ub_V$
is defined and lies in $\calm _{\Q+}^Q(\subS)$.
\end{proof}

{\it  In the subsequent examples, we implicitly fix an orthonormal basis $\mathcal{E}$ with respect to an inner product $Q$ in $\R^\infty$, only referring to these choices in the notation when necessary.}

\begin{ex}  The set ${\calf}^{\rm Ch}=\calf^{{\rm Ch},Q,\mathcal{E}}$ of Chen fractions in Example~\mref{ex:chen} generates the unitary \qorth subalgebra $\calm _\Q^{\rm Ch}\coloneqq \calm^Q_{\Q +}\left(\Pi ^Q({\mti{\calf} }^{\rm Ch})\right)$.
\mlabel{ex:chen1}
\end{ex}

\begin{ex}\label{ex:feynman} For a finite subset $J$ of $\Z_{>0}$,  we set $z_J\coloneqq \sum_{i\in J} z_i$.
As considered by Speer~\mcite{Sp2} (see \S~\mref{ss:ge}), a {\bf Feynman fraction} is a simplex fraction
\vspace{-.1cm}	
\begin{equation}\label{eq:Feynmanfrac}\frac{1}{\prod_{J\in \mathcal{J}}z_J^{s_J}}, s_J>0,
\vspace{-.1cm}	
\end{equation}
for a finite collection $\mathcal{J}$ of finite subsets of $\N$. With similar notations of Example \ref{ex:chen1}, the set  of Feynman fractions and the \qorth subalgebra it generates are denoted by
\vspace{-.1cm}	
$${\calf }^{\rm Fe}
\coloneqq \calf^{{\rm Fe},Q,\mathcal{E}},
\quad \calm _\Q ^{\rm Fe}\coloneqq \calm_{\Q+}^Q\left(\Pi ^Q({\mti{\calf} }^{\rm Fe})\right).
\vspace{-.2cm}	
$$
\end{ex}

\subsection {Automorphism groups of simplex \qorth algebras } \label{sussec:AutoLoc}

For a locality subalgebra $\cala$ of ${(\calm_\QQ,\perp^Q)}$, let ${\rm Aut}^Q(\cala)$ denote the group of locality automorphisms of $\cala$, following Definition~\mref{de:auto}.
\vspace{-.2cm}	

\begin{prop}
\label {pp:homdef}
Let $\cals$ be a set of simplex fractions and let  $\QsubS$
be the unitary \qorth subalgebra of $\Q \mti{\calf}^Q$ generated by $\cals$.
Let ${\rm Aut}_{\rm Res}^Q(\QsubS)$ be the set of unitary  \qorth algebra homomorphisms $\varphi:\QsubS \to \QsubS$ with the property that,
for any fraction $S $ in $\Pi^Q(\cals)$,
\vspace{-.1cm}	
\begin{equation}\mlabel{eq:AutQB}
 \varphi\left(S\right)=
S+ \sum_{i}  {a_i}\, S_i,
\vspace{-.2cm}	
\end{equation}
where
\vspace{-.2cm}	
$$a_i\in \Q, \,S_i \in \Pi ^Q(\mathcal S),  \, \pord(S_i)< \pord(S),\  {\rm Supp}(S_i) \subsetneq {\rm Supp} \left(S\right).
\vspace{-.1cm}	
$$
Then ${\rm Aut}_{\rm Res}^Q(\QsubS)$ is a subgroup of ${\rm Aut}^Q(\QsubS)$.
\end{prop}
\begin{rk}
As a consequence of Theorem \mref{thm:phiiso} yet to come, we have
\vspace{-.1cm}	
 \begin{equation} \label{eq:GalRes}{\rm Aut}_{\rm Res}^Q(\QsubS)  = \Big\{\left.\varphi\in {\rm Aut}^Q(\QsubS) \,\right|\, \varphi \, \text{preserves the p-residue and d-residue}   \Big\}
\vspace{-.1cm}	
 \end{equation}
which justifies the notation with subscript ``Res".
\end{rk}
\begin{proof}
Denote $R\coloneqq R_S\coloneqq \QsubS$.
	We first prove that $\varphi $ is one-to-one.
	Let
\vspace{-.1cm}	
$$f=\sum_i a_i S_i\in R, a_i\in \Q, S_i\in \Pi^Q(\cals)\cup \{1\},
\vspace{-.1cm}	
$$
be nonzero.
	If $f$  is a constant in $\QQ$, then $\varphi(f)=f\neq 0$. If $f$ is not a constant, we group the terms of $f$ according to  the gradation in Eq.~\meqref{eq:ugrad}:
\vspace{-.2cm}	
$$f=c_0+\sum_{U\in \calu} f_U,
\vspace{-.2cm}	
$$
where  $\calu$ is a finite nonempty set of nonzero subspaces of $\RR^\infty$ and $0\neq f_U\in \Q\Pi^Q(\mti{\cals})$, for any $U$ in $\mathcal U$, is a sum  of fractions with  supporting space $U$. Applying $\varphi$ yields
\vspace{-.1cm}	
$$ \varphi(f)=c_0 + \sum_{U\in \calu}\Big( f_U+ \sum_{V\subsetneq U} g_{V}\Big),
\vspace{-.2cm}	
$$
where $g_{V}$ is   a sum (possibly zero) of fractions with supporting space $V$.
For a maximal element $U_0$ in $\calu$, $f_{U_0}$ is the only contribution in the above sum arising in $\varphi(f)$ to the component with supporting space $U_0$ in the decomposition in Eq.~\meqref{eq:ugrad}. So $\varphi(f)=0$ implies $f_{U_0}=0$. This is a contradiction. It follows that $\varphi(f)\neq 0$, which ends the proof of the injectivity.

To prove the surjectivity of $\varphi$, by the linearity of $\varphi$, we only need to show that every element $f$ of $\Pi^Q(\mti{\cals})$ lies in the range Im $\varphi$  of $\varphi$. Suppose this is not the case and let  $U_0\neq 0$ be minimal among the supporting spaces of elements in $ \Pi^Q(\mti{\cals})\backslash \text{Im} \varphi$. Let $f_0$ be one of the simplex fractions in $\Pi^Q(\mti{\cals})\backslash \text{Im} \varphi $  with supporting space $U_0$.
Then by Eq.~\meqref{eq:AutQB} we have
\vspace{-.1cm}	
$$\varphi(f_0) = f_0+ \sum_{V\subsetneq  U_0} f_V,
\vspace{-.1cm}	
$$
where $f_V\in  \Q(\Pi^Q(\mti{\cals}))$ is the sum of simplex fractions with supporting space $V$. The space $U_0$ being minimal, each $f_V$ lies in the image of $\varphi$. Therefore,
$ f_0=\varphi(f_0)-\sum\limits_{V\subsetneq U_0} f_V$
also lies in the image of $\varphi$. This is a contradiction, showing that $\varphi$ is surjective.

We finally prove that $\varphi ^{-1}$ is a \qorth algebra homomorphism. We first show that $\varphi^{-1}$ has the property in Eq.~\meqref{eq:AutQB}, that is, for any $S\in \Pi^Q(\mti{\cals})$, we have
\vspace{-.2cm}	
$$\varphi^{-1} (S)=S + \sum_{i} h_i S_i,
\vspace{-.2cm}	
$$
where each $S_i\in \Pi^Q(\mti{\cals})$ has smaller supporting space and p-order than those of $S$. Assume that this were not the case, and let $S$ have a  minimal supporting space $U_0$ among the counterexamples. Then
\vspace{-.2cm}	
$$ \varphi (S) = S + \sum_{j} b_jT_j,
\vspace{-.2cm}	
$$
where each simplex fraction $T_j\in \Pi^Q(\mti{\cals})$ has it supporting space and p-order smaller than those of $S$. Applying $\varphi^{-1}$ gives
\vspace{-.2cm}	
\begin{equation}	\mlabel{eq:min}
	 S = \varphi^{-1}(S) +\sum_j b_j \varphi^{-1}(T_j).
\vspace{-.2cm}	
\end{equation}
The minimality of the supporting space of $S$ yields
\vspace{-.1cm}	
$$ \varphi^{-1}(T_j)=T_j + \sum_{jk} d_{jk} T_{jk},
\vspace{-.2cm}	
$$
where each $T_{jk}\in \Pi^Q(\mti{\cals})$ has its supporting space and p-order smaller than those of $T_j$ and hence of $S$. Therefore, Eq.~\meqref{eq:min} gives
\vspace{-.2cm}	
$$ \varphi^{-1}(S)=S- \sum_j \varphi^{-1}(T_j)
=S- \sum_j b_j\left (T_j + \sum_{jk} d_{jk} T_{jk} \right),
\vspace{-.2cm}	
$$
which shows that $\varphi^{-1}(S)$ has the form in Eq.~\meqref{eq:AutQB}. This gives the desired contradiction.

To check that $\varphi^{-1}$ is a \qorth map, we consider two linear combinations $f, g$ in $R=\Q(\Pi^Q(\mti{\cals}))$ and group the terms
\vspace{-.1cm}	
$$f= c +\sum_{U\in \calu} f_U \quad {\rm and}
\quad g=d +\sum_{V\in \mathcal V} g_V,
\vspace{-.1cm}	
$$
with $c , d$   in $\Q$ and nonzero sums $f_U, g_V \in \Q(\Pi ^Q(\mti{\cals}))$ of fractions with supporting spaces $U\in \calu$ and $V\in \calv$ respectively.  We proceed to show that $f\perpq  g$ implies $\varphi^{-1}(f)\perpq \varphi^{-1}(g)$.

Let $U$ be an element in $ \calu$ and $V$ an element in $ \mathcal V$.  By Theorem~\mref{thm:poldep}, ${\rm Dep} \,f_U \subset  {\rm Dep} \,f$ and $ {\rm Dep} \,g_V\subset  {\rm Dep} \,g$    so that  $  {\rm Dep}\, f\perpq  {\rm Dep}\, g$ implies ${\rm Dep}\, f_U \perpq  {\rm Dep}\, g_V$   and hence
 $f_U\perpq  g_V.$
By Lemma~\mref{lem:fracdep},
for the linear combination $f_U$ (resp. $g_V$) of simplex fractions with the same supporting space $U$ (resp. $V$), we have $U=\depsp(f_U)$ (resp. $V=\depsp(g_V)$). Thus $f_U\perpq  g_V$ implies $U\perpq  V$.

Since $\varphi^{-1}$ has the property in Eq.~\meqref{eq:AutQB}, we have ${\rm Dep}(\varphi^{-1}(f_U))\subseteq {\rm Dep}(f_U)$ and
${\rm Dep}(\varphi^{-1}(g_V)\subseteq {\rm Dep}(g_V)$. Then
$\varphi^{-1}(f_U)\perpq  \varphi^{-1}(g_V)$.
Therefore, by the linearity of $\varphi^{-1}$, we obtain $\varphi^{-1}(f)\perpq  \varphi^{-1}(g)$, as needed.

Finally for $f,g\in R$ with $f \perp^Q g$, we have $\varphi^{-1}(f)\perp^Q \varphi^{-1}(g)$. Hence
$$\varphi(\varphi^{-1}(f)\, \varphi^{-1}(g))
=\varphi(\varphi^{-1}(f)) \,\varphi(\varphi^{-1}(g)) = fg.$$
Therefore, applying $\varphi^{-1}$, we obtain
$\varphi^{-1}(f)\, \varphi^{-1}(g) = \varphi^{-1}(fg), $
showing that $\varphi^{-1}$ is a locality homomorphism.
\end{proof}

\subsection{\Qorth Galois groups}
\mlabel {subsec:LocGal}
In this part we consider a locality subalgebra $\cala$ of $(\calm_{\Q}, \perp^Q)$ containing $\calm_{\Q+}$. 

\begin{prop}\mlabel{coro:Depphi} For any  \qorth morphism $\varphi: \cala \to \cala$ with $\varphi|_{\calm _{\Q+}}=\Id$, we have
$${\rm Dep} (\varphi(f))\subset {\rm Dep}(f),\quad   \forall f\in \cala.
$$
\end{prop}

\begin {proof} For any $\ell$ in $\call _\Q$ viewed as an element of $(\R^k)^*$ for some $k\geq 1$, if $\ell \perpq  {\rm Dep} (f)$, then $\ell \perpq  f$, which implies that
$\varphi(\ell)\perpq  \varphi(f)
$
since $\varphi$ is a \qorth map.
Since $\ell$ is in $\calm_{\Q+}$, we have $\varphi(\ell)=\ell$.
Thus
$\ell \perpq  {\rm Dep }(\varphi(x))$. So ${\rm Dep}(f)^\top \subseteq {\rm Dep}(\varphi(x))^\top$ which
yields the statement.
\end{proof}

In the sequel, we fix a set $\cals\subseteq \calf$ of simplex fractions and let
\vspace{-.1cm}	
\begin{equation} \mlabel{eq:ab}
\calb\coloneqq \calb(\cals)\coloneqq  \Q(\Pi^Q(\mti{\cals})), \quad \cala\coloneqq \cala(\cals)\coloneqq  \calm_{\Q+}^Q\left(\Pi^Q(\mti{\cals})\right)
\vspace{-.1cm}	
\end{equation}
be the unitary locality subalgebra and $\calm_{\Q+}^Q$-subalgebra of $\calm_\Q$ generated by $\cals$, defined in Propositions~\mref{prop:genalg} and~\mref{prop:genalg2} respectively.

\begin{defn} Define a subset of $\Aut^Q(\cala)$ by
	\vspace{-.1cm}	
	\[{ \rm  Gal}^Q(\cala /\calm_{\Q+})\coloneqq \left\{\varphi\in { \rm  Aut}^Q(\cala)\,
	\left |\, \begin{array}{l} \varphi|_{\calm _{\Q+}}=\Id \\ 
		\varphi \, \text{preserves the p-residue, d-residue}\\
			 \text{and the locality subalgebra}\, \calb \end{array}  \right . \right\}
	\vspace{-.1cm}	
	\]
It will be called the {\bf locality Galois group} of $\cala$ over $\calm_{\Q+}$, thanks to Theorem~\mref{thm:phiiso}. 
	\mlabel{defn:Galoisgroup}
\end{defn}
\vspace{-.3cm}	
We next give a locality tensor product property of $\calm_\Q$.
\begin{prop} \mlabel{pp:tensor} Let $\cals\subseteq \calf$ and, as in Eq.~\meqref{eq:ab}, define
 \vspace{-.1cm}	
	\begin{equation} \mlabel{eq:ab2} \notag
		\calb\coloneqq \calb(\cals)\coloneqq  \Q(\Pi^Q(\mti{\cals})), \quad \cala\coloneqq \cala(\cals)\coloneqq  \calm_{\Q+}\left(\Pi^Q(\mti{\cals})\right).
\vspace{-.1cm}	
\end{equation}
	\begin{enumerate}
		\item \mlabel{it:locten}
		For each $\QQ$-subspace $U$ of $\R^\infty$, let
		$\cala_U\coloneqq \cala\cap \calm_{\QQ,U}$ and let $\calb_U$ denote the linear span of simplex fractions in $\calb$ with supporting space $U$. Let $\calm_{\QQ+}^U$ denote the space of holomorphic germs whose dependent space is contained in
		$$U^{\perpq }\coloneqq \left\{y\in \call(\C^\infty)\,\left|\, y\perpq  u, \forall u\in U \right. \right \}.$$
		Then we have the (inner) tensor product
		$$ \cala_U= \calm_{\QQ+}^U \otimes \calb_U,$$
that is,  $\calm_{\QQ+}^U$ and $\calb_U$ are linearly disjoint.		
\item \mlabel{it:linext}
		Let $(V,\top)$ be a locality vector space. Any pair of locality linear maps $\varphi:\calb\to V$ and $ \psi:\calm_{\Q+}\to V$ uniquely extends to a locality linear map
		$$\varphi\ot^Q \psi: \cala \to V.$$
	\end{enumerate}
\end{prop}

\begin{proof}
	\meqref{it:locten} By Proposition~\mref{prop:genalg2},
\vspace{-.2cm}	
\begin{equation} \mlabel{eq:expand}
	f= \sum_i h_i S_i\,,
\vspace{-.2cm}
\end{equation}
where $S_i\in \subS$ with $\supp(S_i)=U$, $h_i$ is holomorphic with dependent space contained in $U^{\perpq }$ and hence $h_i\perpq S_i$. So $\cala_U= \calm_{\QQ+}^U  \calb_U$ as a product of subsets.

To prove the disjointness, we more generally consider a linear combination
\vspace{-.1cm}	
\begin{equation} \mlabel{eq:exp2}
	\sum_i h_i S_i=0,
\vspace{-.3cm}
\end{equation}
where $\supp(S_i)=U$, $h_i$ is holomorphic with dependent space contained in $U^{\perpq }$. Suppose that $\{S_i\}_i$ is linearly independent, but $h_i\neq 0$ for all $i$ in the sum. Denote $V=\sum_i \supp(h_i)$ which is a finite-dimensional subspace of $U^{\perpq }$. Then $h_i$ is defined on $V$. Thus we can choose disjoint sets of variables $\{z_k\}$ of $U$ and $\{w_\ell\}$ of $V$ respectively. From $h_i\neq 0$, there is $\{w_{\ell}^0\}$ such that $h_i(\{w_{\ell}^0\})\neq 0$ for all $i$. Then Eq.~\meqref{eq:exp2} gives
\vspace{-.2cm}	
$$ \sum_i h_i(\{w_\ell^0\})\, S_i =0,
\vspace{-.2cm}	
$$
showing that $\{S_i\}_i$ is linearly dependent. This gives the desired contradiction.

\noindent
\meqref{it:linext}
By Proposition~\mref{prop:genalg2}, $\cala$ is linearly spanned by homogeneous elements with respect to the grading by supporting space in the grading Eq.~\meqref{eq:ugrad}. Thus $\cala$ has the restricted grading
$$\cala=\bigoplus_{U \subset \R^\infty } \cala_U.
 $$
Then we just need to show that $\varphi$ and $\psi$ uniquely define a locality linear map
	$$ (\varphi\ot^Q \psi)_U: \cala_U\to V
	\vspace{-.3cm}
	$$
	for each subspace $U$ of $\R^\infty$.
	
By Item~\meqref{it:locten}, for any linear map $\varphi:\calb_U\to V$ and $\psi:\calm_{\QQ+}^U:\to V$, there is a unique linear map
	\begin{equation}
		\mlabel{eq:exten2}
(\varphi\ot^Q \psi)_U:\cala_U\to V,\quad f\mapsto \sum_i \psi(h_i) \varphi(S_i)
\vspace{-.3cm}
	\end{equation}
for any element $f=\sum_i h_iS_i$ in $\cala_U$, expressed in the form in Eq.~\meqref{eq:expand}.
Indeed, $(\varphi\ot^Q\psi)_U$ is simply the tensor product of the restriction of $\psi$ to $\calm_{\QQ+}^U$ and the restriction of $\varphi$ to $\calb_U$. Taking the sum over all subspaces $U$ of $\call(\R^\infty)$ including $U=0$, we have an extension $\varphi\ot^Q\psi$ of $\varphi$ and $\psi$ to $\cala$.
\end{proof}

The remaining part of the section is devoted to the proof of the following theorem which extends an element of   ${\rm Aut}_{\rm Res}^Q(\calb )$ to an element of  $\Gal^Q(\cala/\calm_{\Q+})$.

\begin{thm}  \mlabel{thm:phiiso}
Let $\cals\subseteq \calf$ and let $\cala$ and $\calb$ be as defined in Eq.~\meqref{eq:ab}.
\begin{enumerate}	
\item
Any element $\varphi\in {\rm Aut}_{{\rm Res}}^Q(\calb)$  (see Proposition~\mref{pp:homdef}) uniquely extends to
 an  element of  $\Gal^Q(\cala/\calm_{\Q+})$ defined by
\vspace{-.2cm}	
	\begin{equation}\tilde \varphi\left(\sum_i h_iS_i\right)\coloneqq \sum_i h_i \varphi(S_i)
\mlabel{eq:varphiextended}
\vspace{-.2cm}	
	\end{equation}
for
\vspace{-.2cm}	
\begin{equation}
f=\sum_i h_i S_i\in \cala, h_i\in \calm_{\Q+}, S_i\in \Pi^Q(\mti{\cals})
\mlabel{eq:fexpn}
\vspace{-.2cm}	
\end{equation}
as in Proposition~\mref{prop:genalg2}.
\mlabel{it:phi1}
 \item
\mlabel{it:phi2}
The subset $\Gal^Q(\cala /\calm_{\Q+})\subseteq \Aut^Q(\cala)$ is a subgroup. Restricting to $\calb$ gives rise to a group isomorphism
\vspace{-.2cm}	
$$\Gal^Q(\cala /\calm_{\Q+}) \cong {\rm Aut}^Q_{{\rm Res}}(\calb).
\vspace{-.2cm}	
$$
\end{enumerate}
\end{thm}

\begin{proof}
\meqref{it:phi1}
Applying Proposition~\mref{pp:tensor} with $\psi$   the identity map, for any linear map $\varphi:\calb_U\to \calb_U$, there is a unique linear map
\vspace{-.2cm}	
\begin{equation}
 \mlabel{eq:exten3}
\tilde{\varphi}:\calm_{\QQ,U}\to \calm_{\QQ,U}, \quad  f\mapsto \sum_i h_i \varphi(S_i)
\vspace{-.2cm}	
\end{equation}
for any element $f=\sum_i h_iS_i$ in $\calm_{\QQ,U}$, expressed in the form in Eq.~\meqref{eq:expand}.

We next show that the $\tilde{\varphi}$ obtained this way has the form in Eq.~\meqref{eq:varphiextended}.
Let $f=\sum_i h_i S_i$ as in Eq.~\meqref{eq:fexpn}. By grouping the terms according to the supporting spaces of $S_i$ as in Eq.~\meqref{eq:ugrad}, we have
$$ f=\sum_{U\in\calu} f_U\ \text{ with }\ f_U=  \sum_{j} a_{Uj} S_{Uj},$$
where $\calu$ is a set of subspaces $U$ of $\R^\infty$ for which $f_U\neq 0$, and for each $U\in \calu$, we have
$$S_{Uj}\in \Pi^Q(\cals), \quad \depsp(S_{Uj})=\supsp(S_{Uj})=U, \quad \depsp(a_{Uj})\perpq  \depsp(S_{Uj})$$
and each term $a_{Uj}S_{Uj}$ is one of the terms in $f=\sum_i h_i S_i$. Thus $f_U$ is in $\calm_{\QQ,U}$ and  we can apply Eq.~\meqref{eq:exten3} and obtain
\vspace{-.2cm}	
$$ \tilde{\varphi}(f_U) = \sum_i a_{Uj}\varphi(S_{Uj}),
\vspace{-.1cm}	
$$
which takes the form in Eq.~\meqref{eq:varphiextended}.
Hence so is $\tilde{\varphi}(f)$. This is what we want.

The fact that $\tilde \varphi$ preserves the p-residue in Eq.~\meqref{eq:plt} and the d-residue in Eq.~\meqref{eq:dlt} follows from the definition and the special form of $\varphi$ on $\calb$.

For $f=\sum_{i} h_iS_i\in \cala=\calm_{\Q+}(\Pi^Q(\cals))$ with $h_i\in \calm_{\Q +}, S_i\in \cals$, the p-residue $\pres(f)$ of $f$ is of the form
$\sum_i^\prime h_i (0) S_i$ where the sum is over simplex fractions $S_i$ in $f$ with the highest order. By the definition of $\tilde{\varphi}$, the sum $\sum_i^\prime h_i S_i$ is still the part of $\tilde{\varphi}(f)$ with the highest order. Therefore,
$\pres(\tilde{\varphi}(f))=\pres(f)$.

The same argument, applied to the dimensions of supporting spaces of the polar germs, shows that $\tilde{\varphi}$ preserves the d-residues.

We next check that $\widetilde{\varphi}$ is a \qorth $\calm_{\Q+}$-algebra homomorphism. For $a, b \in \cala$ with $a\perpq  b$, as in the proof of Proposition \mref {prop:genalg}, we can write them as
\vspace{-.1cm}	
$$a=\sum _U\sum _ih_{U_i}S_{U_i}, \quad
b=\sum _V\sum _j g_{V_j}T_{V_j}
\vspace{-.2cm}	
$$
such that
\vspace{-.2cm}	
$$h_{U_i} \not =0, g_{V_j} \not =0, \ \{h_{U_i}, S_{U_i}\} \perp  \{g_{V_j}, T_{V_j}\}.
\vspace{-.1cm}	
$$
By the special form of $\varphi$, we have
\vspace{-.1cm}	
$${\rm Dep }(\varphi (S_{U_i}))=U,\quad {\rm Dep }(\varphi (T_{Vi}))=V.$$
Then it follows from the definition of $\tilde{\varphi}$ that
$\tilde \varphi (a)\perpq  \tilde \varphi (b).
$
Moreover (treating $h_0$ as $h_{U}S_{U}$ for $U=0$ and the same for $g_0$),
\vspace{-.1cm}	
{\small
$$\tilde \varphi (ab)=
\tilde{\varphi}\left ( \sum _{U,V}h_{U_i}g_{V_j}S_{U_i} T_{V_j}\right)
= \sum_{U,V} h_{U_i}g_{V_j} \varphi(S_{U_i}T_{V_j})
= \sum_{U,V} h_{U_i}g_{V_j} \varphi(S_{U_i})\varphi(T_{V_j})
= \tilde \varphi (a)\tilde \varphi (b).
\vspace{-.1cm}	
$$
}

By construction, the extension $\varphi \mapsto \tilde \varphi$ is functorial:
$$\widetilde {\varphi \psi}=\tilde \varphi \tilde \psi,
\quad \widetilde \id _\calb =\id _\cala.
$$
So for any $\varphi \in {\rm Aut}_{{\rm Res}}^Q(\calb)$, $\tilde \varphi $ is a linear bijection.
The functorial property also shows that $\tilde{\varphi}\,\widetilde{\varphi^{-1}}=\id_\cala$. So $\tilde \varphi ^{-1}=\widetilde{\varphi^{-1}}$. Thus $\widetilde{\varphi}^{-1}$ is also a \qorth $\calm_{\Q+}$-algebra homomorphism. Thus $\tilde{\varphi}$ is in $\Gal^Q(\cala /\calm_{\Q+})$ for all $\varphi\in {\rm Aut}_{{\rm Res}}^Q(\calb)$.
By
$\widetilde{\varphi\,\psi}=\tilde{\varphi}\tilde{\psi}, \quad \tilde{\varphi}^{-1}=\widetilde{\varphi^{-1}},$
the image of the map
\vspace{-.1cm}	
\begin{equation}
\Psi: {\rm Aut}^Q_{{\rm Res}}(\calb)\to  \Gal^Q(\cala /\calm_{\Q+}): \varphi \mapsto \tilde \varphi,
\mlabel{eq:psi}
\vspace{-.1cm}	
\end{equation}
is a subgroup of $\Aut^Q(\cala)$.

\meqref{it:phi2}
The map $\Psi$ defined in Eq.~\meqref{eq:psi} is clearly injective and its  image is in $\Gal^Q(\cala /\calm_{\Q+})$.
Now for any $g \in \Gal^Q(\cala /\calm_{\Q+})$, since it preserves $\calb$, for any $S\in \Pi ^Q (\cals)$,
\vspace{-.2cm}	
 $$g(S)=\sum_i a_iS_i,
\vspace{-.2cm}	
$$
with $a_i\in \Q$, $S_i\in \Pi ^Q(\cals)$.

Note that the p-residue of $g(S)$ equals to the p-residue of $S$ which is just $S$. So we can write
{\small
$$g(S)=\sum a_iS_i=\sum _{{\pord (S_i)}={\pord}(S)}a_iS_i+\sum _{{\pord (S_i)}<{\pord}(S)}a_iS_i
$$
}
and
{\small  $$\sum _{{\pord (S_i)}={\pord}(S)}a_iS_i=S.
$$}
So
{\small $$g(S)=S+\sum _{{\pord (S_i)}<{\pord}(S)}a_iS_i.
$$}

Also, the d-residue of $S$ is also $S$. Write
$$ g(S)=S+\sum_U S_U=S+\sum_{U} \sum_{i} a_{U_i}S_{U_i},$$
with $\dim U<\dim \supsp(S)$, $0\neq S_U=\sum_{i} a_{U_i}S_{U_i}\in \Q(\Pi^Q(\cals))$ and $\supsp(S_{U_i})=U.$
By Lemma~\mref{lem:fracdep}, Proposition~\mref{coro:Depphi} and  Theorem~\mref{thm:poldep}, we obtain
$$U=\depsp(S_U)\subseteq \depsp(g(S))=\depsp(S)=\supsp(S).$$
Hence $\supsp(S_{U_i})\subsetneq \supsp(S)$.
Thus $g|_\calb$ is in ${\rm Aut}^Q_{{\rm Res}}(\calb)$
 and we have
$\Psi (g|_\calb)=g,$
giving us the desired isomorphism.
\end{proof}
\vspace{-.4cm}	

\section{Locality polynomial algebras generated by Chen and Speer fractions}
\mlabel{sec:locpoly}
In this section, we show that the \qorth subalgebras of $\calf$ generated by Chen fractions of Example~\ref{ex:chen} and by the more general class of \thing fractions over $\Q$ are both \qorth polynomial algebras, defined to be the following locality version of polynomial algebras.

\begin{defn}\label{defn:localg}
	Let $(A,\top)$ be a locality $K$-algebra. Let $X$ be a subset of $A$.
	\begin{enumerate}	
		\item A {\bf locality monomial} built from $X$ is a product $x_1\cdots x_r$ where $x_i\top x_j$ for $i\neq j$.
		\item The set $X$ is called {\bf locality algebraically independent} if distinct  locality monomials built from $X$ are linearly independent.
		\item The set $X$ is called a {\bf locality generating set} of $(A,\top)$ if the only locality subalgebra of $(A,\top)$ containing $X$ is $A$ itself.
		\item
		The locality algebra $(A,\top)$ is called a {\bf locality polynomial algebra} generated by $X$ if $X$ is locality algebraically independent, and is a locality generating set of $(A,\top)$. Then we denote
		$ A\cong \bfk_\sloc[X].$
\end{enumerate}
\end{defn}

\subsection {Locality shuffle algebras as locality polynomial algebras}
\mlabel{sec:localityLydonwords}
\mlabel{ss:shuf}
Let us first recall how shuffle algebras  can be viewed as  polynomially generated algebras. For a commutative ring $K$ and a set $X$, let $A=\bfk X$ be the linear space with a basis $X$, $W(X)$ be the set of words with letters in $X$ including the empty word.
The shuffle algebra on $X$ is the space
$$ \Sh(\bfk X): = \bfk W(X),$$
equipped with the shuffle product $\ssha$. More precisely,
$$
(\alpha_1\vec{\alpha}') \ssha (\beta_1 \vec{\beta}')
\coloneqq  \alpha_1 (\vec{\alpha}' \ssha (\beta_1 \vec{\beta}'))
+ \beta_1 ((\alpha_1\vec{\alpha}')\ssha \vec{\beta}'),
\quad \alpha_1,\beta_1\in X, \vec{\alpha}',\vec{\beta}'\in W(X),
$$
with the initial condition $1\ssha \vec{\alpha}=\vec{\alpha} = \vec{\alpha} \ssha 1.$

For a well-ordered set $(X,\leq)$, equip  $W(X)$ with the lexicographic order $\lexle$. A word $w$ in $W(X)$ is called a {\bf Lyndon word} if $w$ is the smallest among all its rotations. Equivalently, $w$ is Lyndon if $w$ is lexicographically smaller than all of its suffixes: $w=uv$ with $u, v\neq 1$ implies $w\lexle v$.

\begin{theorem} \mcite{CFL,Ra} 	\mlabel{thm:lyn}
	Let $(X,\leq)$ be a well-ordered set.
	\begin{enumerate}
		\item  $($Chen-Fox-Lyndon$)$ Any word $w $ in $W(X)$ has a unique factorisation
		\begin{equation}\mlabel{eq:stfact}
			w=w_1^{i_1}\cdots w_k^{i_k}
		\end{equation}
		where $w_1\lexg \cdots \lexg w_k$ are Lyndon words and $i_1,\ldots, i_k\geq 1$.
		\mlabel{it:lyn1}
		\item 		\mlabel{it:lyn2}	
		$($Radford$)$ Let $\bfk$ be a $\QQ$-algebra. The set $Lyn(X)$ of Lyndon words on $X$ is an algebraically independent generating set of $\Sh(\bfk X)$. Thus $\Sh(\bfk X)\cong \bfk[\Lyn(X)]$ is a polynomial algebra. In fact, with the factorisation $w=w_1^{i_1}\cdots w_k^{i_k}$ in Eq.~\meqref{eq:stfact}, we have
		\begin{equation} \mlabel{eq:lynrec}
			w= \frac{1}{i_1! \cdots i_k!} w_1^{\ssha i_1} \ssha \cdots \ssha w_k^{\ssha i_k} + \text{smaller terms}.
		\end{equation}
	\end{enumerate}
\end{theorem}

From a set $U$ we build a set
$$ \overline{\ug}\coloneqq \{x_0\}\sqcup \{ x_u\ |\ u\in \ug\}.$$
 For a well-ordered set $(U,\leq)$, we then define a well-order $\leq $ on $\overline{U}$ by imposing
\begin{equation} \mlabel{eq:uorder}
x_0<x_u, \ \text{ and } x_u\leq x_v \Leftrightarrow u\leq v, \
 \forall u, v\in U.
\end{equation}

Denote
$$W_1(\overline{\ug})\coloneqq \{1\}\cup \prod _{u\in U}  W(\overline{\ug})x_u \text{ and }    \Sh_1(\bfk \overline{\ug})\coloneqq \bfk W_1(\overline{\ug}) = \bfk \bigoplus \left(\bigoplus_{u\in \ug} \Sh(\bfk \overline{\ug})x_u\right).
$$

Note that $\Sh_1(\bfk\overline{U})$ is closed under the shuffle product, and the factorisation of $w$ in $  W_1(\overline{\ug})$ given in Eq.~\meqref{eq:stfact} has its factors in $W_1(\overline{U})$. Thus we obtain (see also~\cite[\S~3.3.1]{Zh})
\begin{prop}
Let $\Lyn_1(\overline{U})\coloneqq \Lyn(\overline{U})\backslash\{x_0\}$, that is, the set of Lyndon words in $\overline{U}$ that do not end with $x_0$. Then $\Sh_1(\bfk\overline{U})$ is a subalgebra of $\Sh(\bfk \overline{U})$ and is a polynomial algebra generated by $\Lyn_1(\overline{U})$.
\mlabel{pp:1lyn}
\end{prop}

We now extend these constructions to the locality setting.

For a locality set $(X, \top)$,
let $W_\sloc(X)$ denote the subset of $W(X)$ consisting of locality words, namely the words $w=w_1\cdots w_k$ in which $w_i\top w_j, 1\leq i\neq j\leq k$, plus the empty word.
Let
$$\Sh_\sloc(\bfk X)\coloneqq \bfk W_\sloc(X).$$

For $w=w_1\cdots w_k$ and $v=v_1\cdots v_\ell$ in $ W_\sloc(X)$, define
\begin{equation} \mlabel{eq:shufloc}
	w \top v \Longleftrightarrow w_i\top v_j, \forall 1\leq i\leq k, 1\leq j\leq \ell.
\end{equation}
Thus for $w, v$ in $ W_\sloc(X)$ with $w\top v$, the word $wv$   also lies in $W_\sloc(X)$. Since a shuffle of $w$ and $v$ is obtained from $wv$ by permuting the factors and hence   still lies in $W_\sloc(X)$, the shuffle product
$w\ssha v$ lies in $\Sh_\sloc(\bfk X)$. It follows that $\Sh_\sloc(\bfk X)$ is a locality algebra.

For a well-ordered set $(X,\leq)$ equipped with a locality relation $\top$,
the Lyndon words in $W_\sloc(X)$  are called {\bf \loc Lyndon words}. The following statement enhances Theorem \mref{thm:lyn} to a locality setting.

\begin{theorem} Let $(X,\leq)$ be a well-ordered set equipped with a locality relation $\top$.
	\begin{enumerate}
		\item  \text{(Locality Chen-Fox-Lyndon Theorem)} Any word $w$ in $W_\sloc(X)$ has a unique factorisation
		\begin{equation}\mlabel{eq:stfactloc}
			w=w_1^{i_1}\cdots w_k^{i_k}
		\end{equation}
		where $w_1\lexg \cdots \lexg w_k$ are locality Lyndon words and $i_1,\ldots, i_k\geq 1$.
		\mlabel{it:lynloc1}
		\item \text{(Locality Radford Theorem)} Let $\bfk$ be a $\QQ$-algebra. The set $Lyn_\sloc(X)$ of locality Lyndon words on $X$ is a locality algebraically independent generating set $($in the sense of Definition~\ref{defn:localg}$)$ of the locality algebra $\Sh_\sloc(\bfk X)$. Thus
		$\Sh_\sloc(\bfk X)\cong \bfk_\sloc[\Lyn_\sloc(X)]$ is a locality polynomial algebra.
		\mlabel{it:lynloc2}	
	\end{enumerate}
	\mlabel{thm:lynloc}
\end{theorem}

\begin{proof}
	\meqref{it:lynloc1}
In the factorisation $w=w_1^{i_1}\cdots w_k^{i_k}$ of a locality word $w$ into Lyndon words in Theorem~\mref{thm:lyn}\meqref{it:lyn1}, each $w_i$ is still a locality word, giving us the existence of the factorisation in Eq.~\meqref{eq:stfactloc}. The uniqueness of the factorisation follows from the uniqueness of the factorisation in Eq.~\meqref{eq:stfact}.

\smallskip
	
\noindent
\meqref{it:lynloc2}
If $w$ is local, then by Item~\meqref{it:lynloc1} and Theorem~\mref{thm:lyn}\meqref{it:lyn2},
$$w=w_1^{i_1}\ssha w_2^{i_2} \ssha \cdots \ssha w_k^{i_k}+\text{smaller terms},$$
that is, $w$ can be generated by locality Lyndon words modulo smaller terms.
The smaller terms are obtained from $w$ by permuting the letters in $w$, and hence are again local.
Thus as in the nonlocality case, an induction can be applied to show that $\Sh_\sloc(\bfk X)$ is spanned by $\Lyn_\sloc(X)$.
Locality algebraic independence of the set $\Lyn_\sloc(X)$ is automatic since it is a subset of the algebraically independent set $\Lyn(X)$ and locality algebraic independence is weaker than algebraic independence. Therefore, locality Lyndon words are locality polynomial generators of $\Sh_\sloc(\bfk X)$.
\end{proof}

Now let $(U,\leq)$ be a well-ordered set equipped with an irreflexive locality relation $\perp$, i.e., $u\not \perp u$ for any $u$ in $U$.

Then $(\overline{\ug}, \leq)$ is a well-ordered set by Eq.~\meqref{eq:uorder}, and is equipped with a locality relation $\top$ defined by
\begin{equation}
		x_0\top x_u, \forall u\in U\cup\{x_0\} \text{ and }
		x_u\top x_v \text{ whenever } u\perp v, \quad \forall u, v\in U.
\mlabel{eq:ordloc}
\end{equation}

Denote
$$ W_{1,\sloc}(\overline{U})\coloneqq W_1(\overline{U})\cap W_\sloc(\overline{U}), \quad  \Sh_{1,\sloc}(\bfk\overline{U})\coloneqq \bfk W_{1,\sloc}(\overline{U}).$$

\begin{coro}
\mlabel{co:lynlocu}
Let $(U,\leq )$ be a well-ordered set equipped with an irreflexive locality relation $\perp$.
\begin{enumerate}
	\item  Any word $w$ in $W_\sloc(\overline{U})$ admits a unique factorisation
	\begin{equation}\mlabel{eq:stfactloc2}
	w=w_1\cdots w_k x_0^r,
	\end{equation}
	where $w_1\lexg \cdots \lexg w_k \lexg x_0$ are locality Lyndon words in $W_\sloc(\overline{U})$, and $k, r\geq 0$.
	\mlabel{it:lynlocu1}
	\item 	
\mlabel{it:lynlocu2}	
Given a   $\QQ$-algebra $\bfk$, the locality algebra $\Sh_\sloc(\bfk \overline{U})$ is a locality polynomial algebra generated by the set $\Lyn_\sloc(\overline{U})$ of locality Lyndon words on $\overline{U}:$
$\Sh_\sloc(\bfk \overline{U})\cong \bfk_\sloc[\Lyn_\sloc(\overline{U})].$
\item \mlabel{it:lynlocu3}
The subspace $\Sh_{1,\sloc}(\bfk \overline{U})$ of $\Sh_\sloc(\bfk \overline{U})$ is a locality subalgebra. It is also a locality polynomial algebra generated by $\Lyn_{1,\sloc}(\overline{U})\coloneqq \Lyn_1(\overline{U})\cap \Lyn_\sloc(\overline{U})$.
\end{enumerate}
\end{coro}

\begin{proof}
\meqref{it:lynlocu1}
Since $x_0$ is the smallest locality Lyndon word in $W_\sloc(\overline{U})$, it must appear at the end of factorisation in Eq.~\meqref{eq:stfactloc}, giving us $w=w_1^{i_1}\cdots w_k^{i_k}x_0^r$ with $w_1\lexg w_2 \lexg \cdots \lexg w_k\lexg x_0$. Further, a locality Lyndon word $w_i\lexg x_0$ must have a factor $x_u$ for some $u\in U$. So the irreflexivitiy of the locality and the locality of $w$ implies that $i_1=\cdots = i_k=1$ and $w_i\top w_j$ for $1\leq i\neq j\leq k$. The uniqueness of this factorisation follows from the uniqueness of the factorisation in Theorem~\mref{thm:lynloc}.\meqref{it:lynloc2}.

\smallskip

\noindent
\meqref{it:lynlocu2} This follows from Theorem~\mref{thm:lynloc}.\meqref{it:lynloc1}.

\smallskip

\noindent
\meqref{it:lynlocu3} The proof goes as for Proposition~\mref{pp:1lyn}.
\end{proof}

\subsection{Applications  to ordered fractions}
\mlabel{ss:fraceval}
For the rest of this section, we fix an orthonormal basis $\cale$ of $\RR^\infty$ with respect to the inner product $Q$. For the sake of simplicity, this choice, which we make once for all, is suppressed in most of the notations.

Let $U$ be a countable set and let
\begin{equation}\label{eq:Lmap}L: U\to \call_\QQ, \quad u\mapsto L_u, u\in U,\end{equation}
be a map with values in the space $\call_\QQ$ of linear forms defined as in Eq.~\eqref{eq:Klinearforms} with $K=\QQ$,
which  defines a family of linear forms parameterised by $U$.
For $u_i$ in $ U$, $s_i\geq 1, 1\leq i\leq k$, define the {\bf ordered fraction} (with respect to $L$)
$$ \frakf^\gc \pfpair{s_1,\ldots,s_k}{{u_1},\ldots,{u_k}}\coloneqq \frac{1}{L_{u_1}^{s_1}(L_{u_1}+L_{u_2})^{s_2}\cdots (L_{u_1}+\cdots+L_{u_k})^{s_k}}.$$

Define the set of {\bf \general fractions} (with respect to $L$)
\begin{equation} \mlabel{eq:pafr}
	\calf^{\gc}\coloneqq \bigg\{ \frakf^{\gc}\pfpair{s_1,\ldots,s_k}{{u_1},\ldots,{u_k}}\, \bigg | s_i\geq 1, u_i\in U, 1\leq i\leq k, k\geq 0\bigg\} \subseteq \calf.
\end{equation}
Note that  $\QQ\calf^{\gc}$ is the $\QQ$-subspace spanned by ${\rm im}(L)$.

\begin{ex}
When $L=L_{\wch}: \Z_{>0}\to \call_\Q$ is given by $L(u)=z_u$, then
	 \begin{equation}\label{eq:weqakChen}\frakf^{\gc}\pfpair{s_1,\ldots,s_k}{{u_1},\ldots,{u_k}}=\frakf^{\wch}\pfpair{s_1,\ldots,s_k}{{u_1},\ldots,{u_k}}\end{equation}
	from Example~\mref{ex:chen1} and $\calf^L=\calf^{\wch}$ is the set of \weak Chen fractions (called MZV fractions in~\mcite{GX} for applications to multiple zeta values).

Note that in contrast with a weak Chen fraction, a Chen fraction requires $u_i\neq u_j$ for $i\neq j$.
On the other hand, the fraction $\frac{1}{(z_1+z_2) (z_2+z_3)}$ is not an ordered fraction with respect to this $L$.
\end{ex}

\begin{prop} \mlabel{pp:phihom}
Let $U$ be countable and let $L:U\to \call_\Q$ be a map. Then the map
\begin{equation}
	\Phi=\Phi^L: (\Sh_1(\QQ \overline{U}),\ssha) \to \QQ\mti{\calf}^{\gc}, \quad
	x_0^{s_1-1}x_{u_1}\cdots x_0^{s_k-1}x_{u_k} \mapsto
	\frakf^{L} \pfpair{s_1,\ldots,s_k}{u_1,\ldots,u_k},
	\mlabel{eq:wordfract}
\end{equation}
is an algebra homomorphism. Here the multiplication on $\QQ\mti{\calf}^{\gc}$ is the natural one in $\calm_\Q$.
\end{prop}
\begin{proof}
When the map $L$ is given by
$$L_\wch: \Z_{>0} \to \call_\Q, \quad u\mapsto z_u,$$ then  $\calf^L=\calf^{\wch}$ with $\calf^{\wch}$ defined in Eq.~\eqref{eq:weqakChen} is the set of MZV fractions, in which  case the conclusion, for $\Phi=\Phi^{\wch}: \Sh_1(\Q \overline{\Z_{>0}})\to  \QQ \mti{\calf}^{\wch} $ follows from ~\cite[Eqs.~(7),(8) and Theorem~2.1]{GX}. For a general $U$ and $L$, on the grounds of  the countability of $U$, we can fix a bijection $\theta: U \to \Z_{>0}$ and thus an algebra isomorphism
$$\theta: \Sh_1(\Q \overline{U})\to  \Sh_1(\Q \overline{\Z_{>0}}).$$
Also note that the change of variables
$ z_i \mapsto L_i, i\in \ZZ_{>0},$
gives rise to an algebra homomorphism
{\small
\begin{eqnarray*}
&\eta:\QQ \mti{\calf}^{\wch} \to \QQ \mti{\calf}^L, & \\
&\frakf^{\wch} \pfpair{s_1,\ldots,s_k}{u_1,\ldots,u_k}
\longmapsto \frakf^\gc \pfpair{s_1,\ldots,s_k}{u_1,\ldots,u_k} =\frac{1}{L_{u_1}^{s_1}(L_{u_1}+L_{u_2})^{s_2}\cdots (L_{u_1}+\cdots+L_{u_k})^{s_k}}.&
\end{eqnarray*}
}
Then $\Phi^\gc$ is just the composition $\eta\circ \Phi^{\wch}\circ \theta$.
\end{proof}

\begin{remark}
The algebra homomorphism $\Phi$ is not injective. For example, $\frakf\pfpair{2}{u_1}=\frac{1}{L_{u_1}^2}$ and $\frakf\pfpair{1,1}{u_1,u_1}=\frac 12 \frac{1}{L_{u_1}^2}$.
Consequently, although the shuffle algebra $\Sh_1(\QQ \overline{U})$ is a polynomial algebra on the Lyndon words, the same cannot be said of $\QQ\widehat{\calf}^{\gc}$. As we will see below, this defect can be remedied under a  locality condition.
\end{remark}

Let $(U,\top)$ be a countable locality set and let $L:U\to \call_\Q$ be as above. Consider the subset
\begin{equation}
\calf^{\gc}_\sloc\coloneqq  \left\{\left . \frakf^\gc\pfpair{s_1,\ldots,s_k}{u_1,\ldots,u_k}\ \right|\ s_i\geq 1, u_i\in \ug, u_i\top u_j \in U, 1\leq i\not =j\leq k, k\geq 0\right\} \subseteq \calf^L
	\mlabel{eq:pfsp}
\end{equation}
and the $\QQ$-subspace $\QQ \calf^{\gc}_\sloc$ of $\QQ\calf^L$.

\begin {thm} \mlabel{thm:conefrac}
Let $(U,\leq,\top)$ be a countable well-ordered set with an irreflexive locality relation $\perp$. Suppose that the map  $L: (U,\top)\to (\call_\QQ,\perp^Q)$ defined in Eq. \eqref{eq:Lmap} is a locality map in the sense of Eq.~\meqref{eq:locmap}$:$
 for $x,y\in U$, if $x\top y$, then $L_x\perp ^Q L_y.$ Then
\begin{enumerate}
\item
the set $\calf^{\gc}_\sloc$ is linearly independent;
\mlabel{it:indep}
\item
\mlabel{it:shfriso}
the algebra homomorphism $\Phi$ in Eq.~\meqref{eq:wordfract} restricts to an isomorphism of locality algebras
\vspace{-.2cm}	
\begin{equation}
\mlabel{eq:shfr}
\Phi_\sloc: \Sh_{1,\sloc}(\overline{U}) \cong \Q \calf ^{\gc}_\sloc,
\quad w=x_0^{s_1-1}x_{u_1}\cdots x_0^{s_k-1}x_{u_k}\mapsto \frakf^\gc\pfpair{s_1,\ldots,s_k}{u_1,\ldots,u_k}.
\vspace{-.2cm}	
\end{equation}
\item \mlabel{it:chenpoly}
The locality algebra $\QQ \calf ^{\gc}_\sloc$ is a locality polynomial algebra.
\end{enumerate}
\end{thm}

\begin{proof}
\meqref{it:indep}
The first assertion follows from the facts that the supporting cones of all ordered  fractions are projectively properly positioned and that  by \cite[Proposition 3.6]{GPZ3} recalled in  Proposition~\mref{pp:proper}, a projectively properly positioned family of simplicial  fractions is linearly independent.

\smallskip
	
\noindent
\meqref{it:shfriso} By the assumption on $(U,\leq,\top)$, a word $w=x_0^{s_1-1}x_{u_1}\cdots x_0^{s_k-1}x_{u_k}$ in $ \Sh(\QQ \overline{U})$ is local if and only if $u_i\top u_j, i\not= j$. So $w$ is local if and only if
$\Phi_\sloc(w)=\frakf^\gc \pfpair{s_1,\ldots,s_k}{u_1,\ldots,u_k}$ lies in $\calf^\gc_\sloc$. Thus by Item~\meqref{it:indep}, $\Phi_\sloc$ sends a linear basis of $\Sh_{1,\sloc}(\QQ \overline{U})$ to a linear basis of $\QQ \calf _\sloc ^L$ and   is therefore a linear isomorphism.
	
The linear map also preserves the locality. Indeed, for $$w_1=x_0^{s_1-1}x_{u_1}\cdots x_0^{s_k-1}x_{u_k},\  w_2=x_0^{t_1-1}x_{v_1}\cdots x_0^{t_\ell-1}x_{v_\ell}\in W_1(\overline{U}),
\vspace{-.2cm}	
$$
\vspace{-.2cm}	
we have
$$ w_1\top w_2 \Leftrightarrow \{u_1,\ldots,u_k\}\top \{v_1,\ldots,v_\ell\}\Leftrightarrow \{L_{u_1},\ldots,L_{u_k}\}\top \{L_{v_1},\ldots,L_{v_\ell}\} $$
$$\Leftrightarrow \supsp(\Phi(w_1))\perp^Q \supsp(\Phi(w_2)) \Leftrightarrow \Phi(w_1)\perp^Q \Phi(w_2).$$
Therefore $\Phi$ restricts to a locality linear bijection
$\Phi_\sloc: \Sh_{1,\sloc}(\Q\overline{U}) \to \Q\calf^\gc_\sloc$.
Finally, the multiplicativity for $\Phi$ in Eq.~\meqref{eq:wordfract} restricts to one for $\Phi_\sloc$. Hence $\Phi_\sloc$ is a locality algebra isomorphism.

\smallskip
	
\noindent
\meqref{it:chenpoly} This last assertion follows from Item~\meqref{it:shfriso} and Corollary~\mref{co:lynlocu}.\meqref{it:lynlocu3}.
\end{proof}

We consider two special instances of maps  $L:U\to \call_\Q$.

\begin{ex} \mlabel{ex:uexam}
\begin{enumerate}
\item \mlabel{it:uex1}
Let $U$ be $\ZZ_{>0}$ equipped with the natural order and the locality relation
$n\top m \Leftrightarrow n\not=m.$
Define
$$L:\ZZ_{>0}\to \call_\QQ , i\mapsto z_i, i\in \ZZ_{>0}.$$
The corresponding set $\calf^\gc_\sloc$ of \general fractions is the set $\calf^{\rm Ch}$ of Chen fractions in Example~\mref{ex:chen}.
\item \mlabel{it:uex2}
Let $U$ be the set $\fpower(\ZZ_{>0})$ of nonempty finite subsets of $\ZZ_{>0}$.
The order is the lexicographic order: for elements
$$I\coloneqq \{i_1>i_2>\cdots>i_r\}, \quad J\coloneqq \{j_1>j_2>\cdots>j_s\}$$
in $\fpower(\ZZ_{>0})$, define $I\geq J$ if either the first nonzero element in the sequence
$$i_1-j_1, i_2-j_2,\ldots, i_{\min\{r,s\}}-j_{\min\{r,s\}}$$
is positive, or the above sequence of numbers are all zero and $r>s$. The locality relation in $\fpower(\ZZ_{>0})$ is:
$$I\top J \Leftrightarrow I\cap J=\emptyset.
$$
Define
\begin{equation} \label{eq:lmap}
L:\fpower(\ZZ_{>0})\to \call _\QQ, \ \ I\mapsto z_I\coloneqq \sum_{i\in I}z_i, \quad \forall I\in \fpower(\ZZ_{>0}).
\end{equation}
The corresponding fractions $\frakf^\gc\pfpair{s_1,\ldots,s_k}{I_1,\ldots,I_k}$ are of the form
\begin{equation}\label{eq:genChen} \frac{1}{z_{I_1}^{s_1}(z_{I_1}+z_{I_2})^{s_2}\cdots
	(z_{I_1}+\cdots+z_{I_k})^{s_k}}, \, \  s_i\in \Z_{>0},k\in \N, I_j\cap I_{j'}=\emptyset   \, {\rm if}\ j\neq j'
\end{equation}
which we will call {\bf \thing  fractions}. Let $\calf^{\tch}$ denote the  set of \thing   fractions and let $$\calm^{\tch}:=\calm_{\QQ+}\left(\Pi^Q(\mti{\calf^{\tch}})\right).$$
\end{enumerate}
\end{ex}

The term Speer fractions is to recognize that such fractions first appeared in Speer's work~\cite{Sp2} in the context  of renormalisation.  Notice that the linear forms in a Speer fraction correspond to faces of a Chen cone. See Section~\mref{sec:app} for details.

\begin{rk}
We note that ${\calm }_\Q^{\rm Fe}\supset{\calm }_\Q^{\rm Sp}$, yet  whether they actually differ is an open question.
\end{rk}

As a direct consequence of Theorems~\mref{thm:conefrac}, we obtain

\begin{coro}
\mlabel {coro:confrac}
The locality algebras $\QQ \calf ^{\rm Ch}_\sloc$ and $\QQ\calf^{\tch}_\sloc$
are locality polynomial algebras on the set of fractions in $\calf^{\rm Ch}$ and $\calf^{\tch}$ respectively corresponding to the locality Lyndon words.
\end{coro}

\section{Locality characters and analytic renormalisation }
\mlabel{sec:app}
In this last section, as an application of our previous results, we address Problems~\mref{it:sp1}-\ref{it:sp3} raised in the introduction on Speer's analytic renormalisation. We first show that the pole structures of the generalised Feynman amplitudes in Speer's analytic renormalisation are of the form introduced earlier. We then compare our constructions of locality generalised evaluators with those of Speer, and obtain a transitivity group action on locality generalised evaluators.

\subsection {Speer's $s$-families and  Speer fractions}
\mlabel{ss:sfam}
In his work on analytic renormalisation~\mcite{Sp2,Sp3,Sp4} (see also~\mcite{BR,DZ}), Speer determined the possible pole structure  of generalised Feynman amplitudes (or regularised Feynman amplitudes). We show that these linear poles are spanned by Speer fractions described in Eq.~\eqref{eq:genChen}.

We first summarise Speer's work, mostly following~\mcite{Sp3}.
A Feynman graph is called {\bf $2$-connected} if it cannot be disconnected by removing a vertex. A family $\tt{E}$ of subgraphs of a Feynman graph $G$ is called a {\bf singularity family} or simply an {\bf $s$-family}~\cite[Definition~2]{Sp3} if
\begin{enumerate}
	\item every element in $\tt{E}$ is either 2-connected or a single line. Let $\tt{E}'$ denote the subset of 2-connected elements in $\tt{E}$;
	\item $\tt {E}$ is nonoverlapping, that is, for $H_1, H_2\in \tt{E}$, either $H_1\subset H_2$ or $H_2\subset H_1$ or $H_1\cap H_2=\emptyset$;
	\item no union of two or more disjoint elements of $\tt{E}$ is 2-connected;
	\item $\tt {E}$ is maximal with these properties.
\end{enumerate}

For a given Feynman graph $G$, the generalised Feynman amplitude $\calt_G$ \cite[Definition 1]{Sp3}) is built from products of propagators assigned to each edge $\ell$ of $G$,  regularised by means of  a complex number  $\lambda _\ell$. It  enjoys a decomposition  \cite[Formula (2.16)]{Sp3} as a sum  over all s-families   $\tt{E}$  of $G$,  of meromorphic functions $\calt_{\tt{E}}$
$$\calt_G=\sum _{\tt {E}}\calt _{\tt{E}},
	$$
	 resulting from writing the closed cone $\prod_{\ell\in E(G)}\{\alpha_\ell\geq 0\}$ as a union $\cup_{\tt{E}}D(\tt{E})$ of closed cones $D(\tt{E})$ associated with each s-family $\tt{E}$ of subgraphs of $G$.

For an s-family $\tt {E}$ of $G$, and $H$ in ${\tt {E}}$, let
$$\Lambda (H)\coloneqq \sum _{\ell\in L(H)}(\lambda _\ell-1),
$$
where $L(H)$ is the set of edges of $H$. Let $\mu (H)$ be the superficial divergence of $H$.
According to Speer~\cite{Sp3} (see Eq.~(2.21), Theorem 3 and the remark that follows, see also~\mcite{Sp4}, Lemma~1.4 and its proof), the possible poles of
$\calt _{\tt {E}}$ are simple poles given by
$$\Lambda (H)-\frac 12 \mu (H)=0, -1, -2, \ldots.
$$

Since we are only renormalising generalised Feynman amplitudes at $\lambda _\ell=1, \ell\in L(G)$, the possible singularities we need to deal with are of the form 	
$ \bigg(\prod_{H\in \tt{E}'} \Lambda(H)\bigg)^{-1}.$
By a change of variables $z_i=\lambda_{\ell_i}-1$, with an ordering $\ell_1,\ldots,\ell_{|L(G)|}$ of $L(G)$, the $\Lambda(H)$ corresponds to the linear form $z_{I(H)}$ in Eq.~\meqref{eq:lmap}, for $I(H)=\{i\,|\,\ell_i\in L(H)\}$.
So we only need to deal with germs of the form
\vspace{-.2cm}	
$$ \bigg(\prod _{H\in \tt {E}'}z_{I(H)}\bigg)^{-1}\, h
\vspace{-.2cm}	
$$
with $h$  a holomorphic germ.

We now address Problem~\mref{it:sp1}.
\begin {prop}\label{prop:FeynmanamplSpeer} For any s-family $\tt {E}$ of $G$, the fraction
\vspace{-.1cm}	
\begin{equation}\label{eq:fractonsfamily}\bigg( \prod _{H\in \tt {E}'}z_{I(H)}\bigg)^{-1}
\vspace{-.1cm}	
\end{equation}
lies in $\Q\calf^{\tch}$. Thus the germs of the generalised Feynman amplitudes at $z=0$ are in $\calm^{\tch}_\Q$.
\end{prop}
\begin {proof} Since there is no overlaps between any two 2-connected subgraphs in an s-family $\tt{E}$, the Hasse diagram of $\tt{E}'$ is a rooted forest. The flattening procedure used in \cite[Theorem 5.11] {CGPZ2} which involves the flattening morphism defined in \cite[Definition 2.10]{CGPZ2}, transforms this rooted forest into a linear combination of ladder trees.  Correspondingly, the fraction in Eq.~\eqref{eq:fractonsfamily} is a linear combination of (ordered) Speer fractions.
\end {proof}
\subsection{Generalised evaluators on \qorth subalgebras of meromorphic germs}
\mlabel{ss:ge}
 Let $\cala$ be a locality subalgebra of the algebra $\calm_\Q$ equipped with the locality relation $\perp^Q$ of Definition~\ref{defn:depspace}.

\begin{defn} A {\bf \qorth generalised evaluator} $\cale $ on the locality algebra $(\cala,\perp^Q)$    is  a linear  form
	$\cale: \cala \to \C,$
	such that
	\begin{enumerate}
		\item $\cale(h)=  h(0)$ for $h \in \mathcal M_{\Q +}$.
		\mlabel{it:geneval1}
		\item  $\cale(f_1\cdot f_2)=  \cale(f_1)\cdot \cale(f_2)$ for $f_1, f_2\in \cala$ with $f_1\perpq  f_2$.
		\mlabel{it:geneval2}
	\end{enumerate}
	We use $E(\cala)=E^Q(\cala)$ to denote the set of locality generalised evaluators on $(\cala,\perp^Q)$.
	\label{defn:geneval}
\end{defn}

\begin{rk}\mlabel{rk:locfil}
Notice that a locality subalgebra $\cala\subset \calm_\Q$ is defined as a direct limit $\cala=\underset{\longrightarrow}{\lim}\cala_k$ with $\cala_k\subset \calm_\Q  (\C^k)$, so that a linear form $\cale:\cala\to \C$ amounts to a family of linear forms $\cale_k: \cala_k\to \C, k\in \N$ such that	 $$\cale_k\vert_{ \cala_{k-1}}=\cale_{k-1}, \quad k\geq 1.$$ 	
\end{rk}

\begin{ex}
	Let ${\rm ev}_0: \calm_{\Q+}\to \Q$ be the evaluation at $0$ defined as ${\rm ev}_0(h)=h(0)$.  Using the map
	$\pi_+^Q$ defined in Eq.~\meqref{eq:piqplus}, we build a \qorth generalised evaluator on $\cala$ as the composition
	\begin{equation}\mlabel{eq:mseva}
	\mseva \coloneqq {\rm ev}_0\circ  \pi_+^Q,
	\end{equation}
	which we call  the {\bf \qorth minimal subtraction map} since it is a generalisation to  higher  dimensions of  the one variable minimal subtraction map obtained by projection onto the holomorphic part of the Laurent series ring $\C[\vep^{-1},\vep]]$ followed by evaluation at zero.
\end{ex}

In his work~\mcite{Sp1,Sp2,Sp3,Sp4}, Speer achieved analytic renormalisation by means  of linear forms he called generalised evaluators. Let us compare these generalised evaluators with those defined in Definition~\mref{defn:geneval}.

Let us first observe that, as in Example~\mref{ex:chen}, a function is given in the variables $z_i$  means  that we have chosen a basis $\{e_i\}$ of the space, and every element in the space is written in the form $\sum z_i e_i$. In the presence of an inner product we choose $\{e_i\}$ to be an orthonormal basis.

Following Speer (and see Example~\mref{ex:feynman}), define $\calm^{\rm Fe}(\C^k)$ to be the spaces of meromorphic germs (at zero) $f$ with the property that
\vspace{-.1cm}	
$$ f:\C^k\to \C \ \text{ such that } \   f(z_1, \ldots, z_k)\cdot   \prod_{I\subseteq [k]}  z_I  \, \text{ is holomorphic at zero}.
\vspace{-.2cm}	
$$
Speer defines a {\bf generalised evaluator}~\mcite{Sp1}  as a family of linear maps
$$\cale\coloneqq \Big\{ \cale_k:\calm^{\rm Fe}(\C^k)\to \C\Big\}_{k\in \N} $$
satisfying the following conditions.
\begin{enumerate}
\item \mlabel{it:fil} (compatibility with the filtration)
$\cale_k\vert_{ \calm^{\rm Fe}(\C^{k-1})}=\cale_{k-1}, k\geq 1;$	
\item  \mlabel{it:ext} (extension of ${\rm ev}_0$)   $\cale$   is the usual evaluation ${\rm ev}_0$ at zero on   holomorphic functions;
\item  \mlabel{it:mult}  (partial multiplicativity)  $     \cale(f_1\cdot f_2)= \cale(f_1 )\cdot   \cale(  f_2)$ if $f_1$ and $f_2$ depend on disjoint sets of variables $z_i$;
\item \mlabel{it:inv} ($\Sigma_k$-invariance) $\cale$ is invariant under permutations of the variables  $\cale_k\circ \sigma^*= \cale_k$ for any $\sigma\in \Sigma_k$, with $\sigma^* f(z_1, \ldots, z_k)\coloneqq  f(z_{\sigma(1)}, \ldots, z_{\sigma(k)})$;
\end{enumerate}
There are also reality and continuity conditions which we do not discuss here. Continuity of generalised evaluators requires enhancing the constructions carried out here to a topological setting, which is the object of a joint work \cite{DPS} of the second author.

In practice, Speer builds such a generalised evaluator
by setting
\vspace{-.2cm}	
\begin{equation}\label{eq:Eiter}\cale_k^{\rm iter}\coloneqq  \frac{1}{k!}\,\sum_{\sigma\in \Sigma_k}{\rm ev}^{{\rm reg}, z_{\sigma(1)}}_0\circ \cdots \circ{\rm ev}^{{\rm reg},z_{\sigma(k)}}_0,
\vspace{-.2cm}	
\end{equation}
where, for $1\leq i\leq k$, ${\rm ev}^{{\rm reg}, z_i}_0(f)$ is defined by ${\rm ev}_0\circ \pi_+^{z_i}(f)$ when viewing  $f$ as a meromorphic function in the variable $z_i$, where $\pi^{z_i}_+$ is defined as in Eq.~\meqref{eq:proj1}.

\begin{ex} \mlabel{ex:sp} We give some examples  in  the case of $k=2$.
\begin{itemize}
	\item[a)] For $f(u,v)= \frac{u}{v} $, $g(u,v)=\left(\frac{u}{v}\right)^2$,   we have $\cale_2^{\rm iter}(f)=\cale_2^{\rm iter}(g)=0$;
	\item[b)] A change of variable $u=z_1-z_2, v=z_1+z_2$ in  $f$ and $g$ gives
\vspace{-.2cm}	
$$\tilde f(z_1, z_2)=\frac{z_1-z_2}{z_1+z_2}, \quad  \tilde g(z_1,z_2)=\left(\frac{z_1-z_2}{z_1+z_2}\right)^2,
\vspace{-.2cm}	
$$
and	we  have
	$\cale_2^{\rm iter}(\tilde f)=0$ whereas $ \cale_2^{\rm iter}(\tilde g)=1$.
\end{itemize}
\end{ex}
We now show that our locality generalised evaluators are the  generalised evaluators \`a  la Speer and  that it provides a useful alternative to the generalised evaluator $\cale^{\rm iter}$ originally used by Speer.

\begin{prop} \mlabel{pp:comeva} Given  any inner product $Q$, a locality generalised evaluator $\cale$ in Definition~\mref{defn:geneval} satisfies  Conditions \meqref{it:fil} -- \meqref{it:inv} defining generalised evaluators \`a la Speer.
\end{prop}
\begin{proof}
Indeed,
conditions~\eqref{it:fil} follows from Remark~\mref{rk:locfil}, and \eqref{it:ext} is the same as condition~\meqref{it:geneval1} in Definition~\mref{defn:geneval}.
The locality multiplicative condition (Definition~\mref{defn:geneval}.\meqref{it:geneval2}) for $\cale$ is stronger than the partial multiplicativity condition~\meqref{it:mult} since disjointness of variables of two functions implies the orthogonality for $Q$ when the variables correspond to   coordinates in an orthonormal basis for $Q$.
Condition~\eqref{it:inv} follows from Theorem 5.1(ii) in \cite {CGPZ3}.
\end{proof}

In order to address Problem~\mref{it:sp2},  combining Corollary~\mref{coro:PolyA} and Proposition~\mref{pp:comeva}, we observe that every  generalised evaluator \`a la Speer factors through the minimal subtraction scheme $\mseva$.

We compare the two generalised evaluators $\cale^{\rm iter}_k$ and $\mseva$.
\begin{enumerate}
\item We first observe a problem in the inductive procedure proposed by Speer since the function $(z_2, \cdots, z_k)\mapsto  {\rm ev}_0^{\rm reg,z_1}(f(z_1, \cdots, z_k))$ might be non-meromorphic, which is an obstacle to implementing the composition ${\rm ev}_0^{\rm reg,z_j}\circ {\rm ev}_0^{\rm reg,z_1}$ on $f$  for $j\neq 1$.

For example, the discontinuity at zero
$${\rm ev}_0^{\rm reg,z_1}(\tilde f)=\left \{\begin {array}{ll}1,&z_2=0,\\-1, &z_2\not =0,\end {array}\right . \,
$$   is an obstacle for the next step which gives the expression
$ {\rm ev}_0^{{\rm reg}, z_2} \circ {\rm ev}_1^{{\rm reg},z_1}(\tilde f)$
arising in the definition $\cale_2^{\rm iter}(\tilde f)$.
This also suggests that there is no natural way to interpret Speer's generalised evaluator as a minimal subtraction scheme in multiple variables.
In contrast, our generalised evaluator $\mseva$ is defined for meromorphic germs with linear poles, so this does not pose a problem.

\item The partial multiplicativity property (ii) for Speer's generalised evaluator applies to a smaller set of pairs of functions than the one allowed by our locality multiplicativity. For $\tilde g_1(z_1,z_2)= (z_1-z_2)^2 $ and  $\tilde g_2(z_1,z_2)= (z_1+z_2)^{-2} $ we have
\vspace{-.1cm}	
$$ \cale_2^{\rm iter} (\tilde g_1\, \tilde g_2)\neq \cale_2^{\rm iter}(\tilde g_1)\cale_2^{\rm iter}(\tilde g_2)
\vspace{-.1cm}	
$$
since $\cale_2^{\rm iter} (\tilde g_1\, \tilde g_2)=1$ and $\cale_2^{\rm iter} (\tilde g_1)=\cale_2^{\rm iter} (\tilde g_2)=0$. Yet the multiplicativity holds for $\mseva$ since the linear forms $z_1-z_2$ and $z_1+z_2$ are orthogonal (as before, the parameters $z_i$'s correspond to  coordinates  in an orthonormal basis  for $Q$) and we have
\vspace{-.1cm}	
$$\mseva(\tilde g_1\, \tilde g_2)=0=\mseva(\tilde g_1)\,\mseva(\tilde g_2).
\vspace{-.1cm}	
$$
\item As illustrated by Example~\mref{ex:sp}, Speer's generalised evaluator $\cale^{\rm iter}_k$ depends on a choice of basis since  $\cale^{\rm iter}_2 (g)\neq  \cale^{\rm iter}_2(\tilde g)$. It is not invariant even under an orthogonal transformation of the variables. In contrast, our generalised evaluator $\mseva$ does not depend on such an orthogonal transformation. For the function $g$ in Example~\mref{ex:sp} with standard basis $z_1,z_2$ under an inner product $Q$, $\mseva(g)=\mseva(\tilde{g})=0$.
\end{enumerate}

To conclude, the above observations speak in favour of  the use of the global minimal subtraction scheme $\mseva$, since it is globally defined on  $\calm_\Q^{\rm Fe}(\C^k)$ and is multiplicative on a large  set of pairs.
\vspace{-.3cm}	

\subsection{Locality Galois actions on generalised evaluators}
We consider the action of the locality Galois group on generalised evaluators.
Clearly, the group $\Gal^Q(\cala /\calm_{\Q+})$ acts on $E^Q(\cala)$:
\vspace{-.1cm}	
\begin{equation}\label{eq:actionGaloisgroup}
E^Q(\cala)\times \Gal^Q(\cala /\calm_{\Q+}) \to E^Q(\cala),\quad  (\cale,g)\mapsto \cale \circ g.
\vspace{-.1cm}	
\end{equation}

The subsequent theorem shows that the automorphism group  ${\rm Aut}^Q(\cala)\cong \Gal^Q(\cala /\calm_{\Q+})$ acts transitively on $E^Q(\cala)$, and thus relates any \qorth generalised evaluator  to the \qorth minimal subtraction scheme $\mseva$. In this respect,  $\Gal^Q(\cala /\calm_{\Q+})$ can be regarded as a renormalisation group on  generalised evaluators.

\begin{thm} 	\mlabel{thm:evaluators}
For  $\cals\subseteq \calf$, we consider  the unitary \qorth subalgebra $\cala\coloneqq \calm_{\Q+}^Q(\Pi ^Q(\mti{\cals}))$  of $\calm_\Q$.
	Suppose that $ \Q\Pi ^Q(\mti{\cals})$ is a \qorth polynomial subalgebra of $ \mti{\calf}$ with a \qorth polynomial basis $S\subset \Pi ^Q(\cals)$. Then every \qorth generalised evaluator $\cale$ on $\cala$ factorises through the minimimal subtraction scheme  $\mseva$,
	that is, there  is some $\tilde \varphi$ in $ \Gal^Q(\cala /\calm_{\Q+})$ such that
	\begin{equation}\label{eq:ClassE}\mseva \circ \tilde \varphi =\cale.\end{equation}
\end{thm}
\begin{proof}
	Let $\cale\in E^Q(\cala)$ be given.
	By assumption, $\calb\coloneqq \Q\Pi ^Q(\mti{\cals})$ is a \qorth polynomial subalgebra   of $\calm_\Q$ with $S\coloneqq \{s_\alpha \}\subset \Pi^Q(\cals)$ as a set of \qorth polynomial generators. Then the map
	\begin{equation}\label{eq:evvarphi} \varphi: S\to \calb, \quad  s_\alpha\mapsto s_\alpha +\cale(s_\alpha)\end{equation}
	extends to a \qorth algebra homomorphism $\varphi$ on $\calb$.
	
	For any $t$ in $\Pi ^Q(\cals)$, we now determine the form of $\varphi (t)$. Write
	$$t=P(S)$$
	as a polynomial in $S$. Then $P$ is a sum of monomials of the form
	$ \Pi_\alpha s_{\alpha}$ with $s_{\alpha}$ all distinct due to the locality algebraic independence. Further, let $W=\supsp(t)$ and $N=\pord(t)$. Then by Lemma~\mref{lem:SepSpace}, we can assume that all these monomials have supporting space $W$ and p-order $N$.
	For each of these monomials, by the locality multilplicativity, we have
	$$\varphi (\Pi_\alpha s_{\alpha})=\Pi_\alpha(s_\alpha +\cale (s_\alpha))=\Pi_\alpha s_\alpha+ {\rm lower\ order\  terms}.
	$$
	Here the lower order terms come from $\Pi_\alpha s_\alpha$
	by replacing one or more factors $s_{\alpha}$ by $\cale (s_\alpha)$, so they indeed have lower p-orders and smaller supporting spaces.
	Since all the  monomials $\Pi_\alpha s_{\alpha}$ have the same supporting space and p-order, we have
	$$\varphi(t)=\varphi(P(S))=P(S)+ {\rm lower\ order\  terms}.$$
	The map  $\varphi$ therefore defines an element in ${\rm Aut}_{{\rm Res}}^Q(\calb)$. The extension $\tilde \varphi \in \Gal^Q(\cala /\calm_{\Q+})$ of $\varphi$ obtained from Theorem~\mref{thm:phiiso} is a \qorth algebra homomorphism. Further we have
	\[ \mseva\circ \tilde{\varphi} (s_\alpha)={\rm ev}_0 \circ \pi_+^Q \circ \tilde \varphi (s_\alpha )=
	{\rm ev}_0\circ \pi_+^Q(s_\alpha+\cale(s_\alpha))
	={\rm ev}_0(\cale(s_\alpha))=\cale(s_\alpha).\]
	Hence, the \qorth algebra homomorphisms $\mseva \circ \tilde \varphi $ and $\cale$ agree on the locality polynomial generating set $S$ of $\cala$. Therefore $\mseva\circ \tilde{\varphi}=\cale$ on $\calb$. Since they also agree on $\calm_{\Q+}$, by Proposition~\mref{pp:tensor}, they agree on $\calm_{\Q+}(\Pi^Q(\cals))$.
\end{proof}

As a direct consequence of Theorem~\mref{thm:conefrac} and Corollary~\mref{coro:confrac}, we obtain  a statement which addresses  Problem \ref{it:sp3}

\begin{coro}
	\mlabel {coro:PolyA}
		The space $E^Q(\calm^{\rm Ch}_\QQ)$ $\big($resp. $E^Q(\calm^{\tch}_\QQ)$$\big)$ of locality generalised evaluators on the locality algebra $\left(\calm ^{{\rm Ch}}_\Q,\perp^Q\right)$ $\Big($resp. $\left(\calm ^{{\tch}}_\Q,\perp^Q\right)\Big)$ is a homogeneous space of $\Gal^Q\left(\calm ^{{\rm Ch}}_\Q /\calm _{\Q+}\right)$ $\Big($resp. $\Gal^Q\left(\calm ^{{\tch}}_\Q /\calm _{\Q+}\right)$$\Big)$. In other words, these groups act transitively on
$E^Q(\calm^{\rm Ch}_\QQ)$ $\big($resp. $E^Q(\calm^{\tch}_\QQ)$$\big)$.
\end{coro}

To finish the paper, we use multiple zeta values to build an example of locality generalized evaluators and elements of the locality Galois group.

Recall that for $s_1,\ldots,s_k$ in $ \ZZ_{>0}$ with $s_1\geq 2$, the multiple zeta value (also called a multizeta value) at $(s_1, \cdots, s_k)$ is
\vspace{-.2cm}
\begin{equation}\label{eq:multizeta} \notag
\zeta\left( s_1, \cdots, s_k\right):= \sum_{n_k>\cdots >n_1\geq 1} n_1^{-s_1}\cdots n_k^{-s_k}=\prod_{i=1}^k\sum_{m_i=1}^\infty( m_1+\cdots + m_i)^{-s_i}.
\end{equation}

For a locality Lyndon word
$x_0^{s_1-1}x_{u_1}\ldots x_0^{s_k-1}x_{u_k}\in \Sh_{1,\sloc}(\overline{U})$ and for the corresponding ``Lyndon Chen fraction"
\begin{equation}\label{eq:lchenfrac} \notag
	\frakf\pfpair{s_1,\ldots,s_k}{{u_1},\ldots,{u_k}}\coloneqq \frac 1{z_{u_1}^{s_1}(z_{u_1}+z_{u_2})^{s_2}\cdots (z_{u_1}+z_{u_2}+\cdots+z_{u_k})^{s_k}}\,, \   u_i, s_i\in \Z_{>0},k\in \N, u_i\neq u_j  \, {\rm if}\ i\neq j,
\end{equation}
in $\calf^{\rm Ch}_\sloc$ (see Example~\mref{ex:uexam}.\meqref{it:uex1}), define
\vspace{-.2cm}
\begin{equation}\mlabel{eq:mzvchar}
{\mathcal E}^\zeta\left(\frakf\pfpair{s_1,\ldots,s_k}{{u_1},\ldots,{u_k}}\right):=\left \{ \begin {array}{ll}\zeta\left( s_1, \cdots, s_k\right), &s_1\ge 2,\\
	0,& s_1=1.\end{array} \right .
\end{equation}
Then by Corollary~\mref{coro:confrac}, this assignment
extends to a unique locality algebra homomorphism
$$ \cale^\zeta: \cala^{\rm Chen}
\to \R.$$
Here $\R$ is equipped with the full locality condition $\R\times \R$, implying the locality of the homomorphism $\cale^\zeta$.

Note that the map on $\cala^{\rm Chen}$ defined by Eq.~\meqref{eq:lchenfrac} for \emph{all} locality Chen fractions is also a locality algebra homomorphism, following~\mcite{IKZ}. It therefore coincides with $\cale^\zeta$.
Hence we conclude that assigning multiple zeta values  to locality Chen fractions as in (\ref{eq:mzvchar})  defines a locality generalized evaluator $\cale^\zeta$.
Then thanks to Theorem 	\ref{thm:evaluators},  there is a transformation $\tilde \varphi$ in the locality Galois group $ \Gal^Q(\cala^{\rm Chen} /\calm_{\Q+})$, such that
$$\mseva \circ \tilde \varphi =\cale^\zeta.$$

\noindent
{\bf Acknowledgments.} The second author is grateful to the Perimeter Institute in Waterloo where she was hosted on an  Emmy Noether fellowship. This research is supported by
the National Natural Science Foundation of China (11890663 and 11821001).

\noindent
{\bf Declaration of interests.} The authors have no conflicts of interest to disclose.

\noindent
{\bf Data availability.} Data sharing is not applicable to this article as no new data were created or analyzed in this study.

\end{document}